\documentclass[preprint,11pt,3p]{elsarticle}

\usepackage{graphicx}
\usepackage{tabularx}
\usepackage{amssymb}
\usepackage{paralist}
\usepackage{amsmath,mathtools}
\usepackage{amsthm}
\usepackage{makecell}
\usepackage{enumitem}
\usepackage{float}
\usepackage{hyperref}
\usepackage{multirow}
\usepackage[table]{xcolor}
\usepackage{pifont}
\usepackage{tikz}
\usepackage{csquotes}
\usepackage{rotating}
\usepackage{booktabs}
\usepackage{array}
\usepackage{longtable}
\usepackage{url}
\usepackage{doi}
\usepackage{subcaption}
\usepackage{bbm}
\usepackage{threeparttable}
\usepackage{nomencl}
\usepackage{arydshln}
\usepackage{algorithmicx,algorithm,algpseudocode}
\usepackage{siunitx}
\usepackage{cleveref}

\biboptions{sort&compress}
\definecolor{mygreen}{rgb}{0,0.6,0}
\algdef{SE}[DOWHILE]{Do}{doWhile}{\algorithmicdo}[1]{\algorithmicwhile\ #1}%
\algrenewcommand\algorithmicrequire{\textbf{Input:}}
\algrenewcommand\algorithmicensure{\textbf{Output:}}
\makenomenclature
\newcommand{\mc}{\mathcal}
\newcommand{\ol}{\overline}
\newcommand{\ul}{\underline}

\allowdisplaybreaks

\newtheorem{lemma}{Lemma}

\newtheorem{prop}{Proposition}

\crefname{prop}{Proposition}{Propositions}
\Crefname{prop}{Proposition}{Propositions}
\crefname{lemma}{Lemma}{Lemmas}
\Crefname{lemma}{Lemma}{Lemmas}
\crefname{assum}{Assumption}{Assumptions}
\Crefname{assum}{Assumption}{Assumptions}
\hypersetup{colorlinks=true, urlcolor=blue, linkcolor=blue, citecolor=blue}

\begin{document}

\begin{frontmatter}

\title{Trajectory-Based Co-Optimization of Arrival Scheduling and Descent Path Design in the Terminal Maneuvering Area}

\author[inst1]{Yutian Pang\corref{cor1}}
\cortext[cor1]{Corresponding author.}
\ead{yutian.pang@austin.utexas.edu}

\author[inst1]{John-Paul Clarke}

\affiliation[inst1]{organization={Department of Aerospace Engineering and Engineering Mechanics},
            addressline={The University of Texas at Austin},
            city={Austin},
            postcode={78712},
            state={TX},
            country={USA}}

\begin{highlights}
\item We propose a four-dimensional trajectory-based optimization framework to automate traffic control in the TRACON area.
\item The algorithm co-optimizes each aircraft's descent design and lateral path extension, with every candidate descent verified in six-degree-of-freedom simulation.
\item The co-optimized continuous descents save 15\% of fleet fuel below saturation, and delayed deceleration saves 23\%.
\item On a common $3.0^\circ$ glideslope angle, delayed deceleration saves 8--18\% relative to an optimized continuous descent.
\item Experiments show that wind can change the fuel burn of an individual arrival by up to 81\%.
\item At the scenario level these aircraft-specific wind effects largely cancel, and wind explains at most 4\% of the variability in fleet fuel.
\end{highlights}

\begin{abstract}
Terminal arrival scheduling and descent procedure design are studied in two largely separate literatures. Scheduling models reduce each aircraft to a travel time and deliver target landing times, and fuel-efficient descent procedures are designed one aircraft at a time with the schedule taken as given, although both decide where an arriving aircraft absorbs delay before final approach. Existing formulations therefore cannot trade a slower, earlier-configuring descent against level track miles, a schedule that is efficient in time can be expensive in fuel, and autonomous or reduced-crew operations will need a single trajectory plan that ground automation and the flight management system both accept. To close this gap, we propose a four-dimensional terminal arrival scheduler that selects each aircraft's lateral path extension, glideslope-capture distance, and flap-deployment trigger speeds in one decision. We evaluate every candidate idle-thrust descent offline with a wind-aware backward plan and a six-degree-of-freedom forward simulation that returns descent time, fuel burn, minimum track length, and stabilized-approach feasibility, and a rolling-horizon scheduler commits one verified descent and one extension per aircraft under wake-separation constraints and observed entry winds. Two Atlanta terminal airspace case studies quantify the benefit. We show that co-optimized continuous descents save about 15\% of fleet fuel below saturation in a free-descent environment and that delayed deceleration saves about 23\%, while on the six published Runway 8L arrival flows the charted altitude floors remove 31\% of the design lattice and reduce the savings to 9--10\% and 20--21\%, respectively. We also find that wind moves single-aircraft descent fuel by 34--81\% yet explains at most 4\% of fleet fuel variance, because aircraft-specific wind effects average out across a scenario.
\end{abstract}


\end{frontmatter}


\section{Introduction}\label{sec:introduction}

Trajectory-based operations (TBO) are the operating concept that air navigation authorities have adopted for the next generation of air traffic management. ICAO developed the Global TBO Concept and integrated it into the Global Air Navigation Plan \cite{icao2018tbo}. In the United States, the NextGen concept of operations placed the four-dimensional trajectory at its center \cite{jpdo2007concept}, the Federal Aviation Administration (FAA) maintains TBO as a NextGen program \cite{faa_tbo_2026}, and NASA developed the Management by Trajectory concept \cite{fernandes2018concept}. In Europe, SESAR and EUROCONTROL flew the first initial 4D demonstration in 2012 and continue the work in the 4D Skyways project \cite{eurocontrol2012i4d,sesar2023skyways}. In China, the Civil Aviation ATM Modernization Strategy sets TBO as one of three major goals, and a dual-aircraft TBO validation flight was completed in 2024 \cite{caac2022tbo,caac2025tbo}. Air navigation service providers from eight Asia-Pacific states, including China, Japan, Singapore, and the United States, signed a TBO Pathfinder agreement in 2023 \cite{caac2025tbo}. Under TBO, each flight is managed by a four-dimensional (4D) trajectory that ground and airborne systems develop together and then hold as a common reference \cite{faa_tbo_2026,icao2018tbo,enea2012comparison,gardi2016multiobjective}. The expected benefits are a more predictable traffic flow, fuel-efficient climbs and descents, lower controller workload, and earlier conflict resolution \cite{torres2012integrated,gardi2016multiobjective}. These benefits depend on the ground and airborne systems holding the same trajectory, and inconsistent air and ground trajectories remain a limitation of current operations \cite{icao2018tbo,klooster2010trajectory,torres2011trajectory,bronsvoort2016role,mondoloni2020aircraft}.

Terminal airspace is where this requirement is hardest to meet. At a busy international hub such as Atlanta, six published arrival flows enter from different directions and must merge into one stream that satisfies wake-turbulence separation at the final approach fix (FAF) \cite{ren2009contrast,balakrishnan2006scheduling,bennell2011airport}. \Cref{fig:ops_a} shows this environment as it appears on a terminal radar display for Runway 8L in east-flow operations. Controllers achieve the merge with speed control, path stretching, and vectoring, issued one clearance at a time. Terminal-area inefficiency accounts for 1.5 to 4.5\% of the fuel burned on a flight in the United States, and anticipated delay burns only 10 to 20\% of the fuel of unanticipated delay \cite{ryerson2014time}. Arrivals at 27 European airports spent an average of 3.4 minutes of additional time inside the 40 nmi arrival area in 2024 \cite{prc2025prr}, and level segments flown in descent cost about ten times the fuel of level segments flown in climb \cite{prc2020vertical}. The penalty grows as demand approaches the runway acceptance rate, because delay is then absorbed in level flight at metering altitude, where fuel per nautical mile is high \cite{nikoleris2016comparison,ng2024optimization}. Clearance-by-clearance control also limits predictability. The ground system and the flight management system (FMS) predict the same flight from different information \cite{mondoloni2020aircraft}, and for a sample optimized profile descent the ground prediction without intent exchange erred by up to 11{,}000 ft in altitude and 70 s in time \cite{konyak2009improving}. A terminal arrival under TBO must resolve both costs where the ground schedule and the airborne descent meet. \Cref{fig:ops_b} previews the alternative developed in this paper, in which each arrival leaves its metering fix with a committed lateral path, descent design, and time at the FAF.

\begin{figure*}[!t]
\centering
\begin{subfigure}[t]{0.492\textwidth}
\centering
\includegraphics[width=\linewidth]{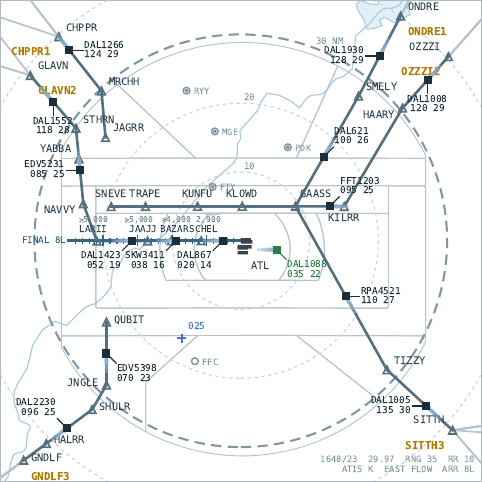}
\caption{The current operation of KATL RW08L charted procedures on an A80 terminal radar display.}
\label{fig:ops_a}
\end{subfigure}\hfill
\begin{subfigure}[t]{0.492\textwidth}
\centering
\includegraphics[width=\linewidth]{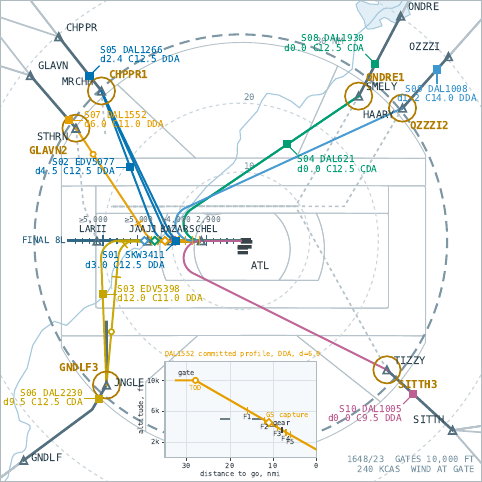}
\caption{A demonstrative view of the proposed 4D trajectory-based arrival manager.}
\label{fig:ops_b}
\end{subfigure}
\caption{Runway 8L arrivals at Atlanta in east-flow operations, rendered in the style of a terminal radar display with true north up and range rings every 10 nmi. In panel (a), the six published RNAV arrival flows merge under radar vectoring onto the common final approach course, which carries the charted altitude floors of the ILS 8L fixes, and data blocks give the callsign, the altitude in hundreds of feet, and the ground speed in tens of knots. In panel (b), each arrival receives at its metering gate a committed four-dimensional trajectory in the color of its flow, with a centerline extension $d_i$, a glideslope-capture distance (diamond), a flap trigger schedule (ticks), and a scheduled time at the FAF, and the inset shows the committed vertical profile of one aircraft. Charted geometry comes from the FAA NASR and CIFP databases, and the traffic is illustrative.}
\label{fig:ops}
\end{figure*}

The case for automating terminal arrival management rests on controller workload and safety. The capacity of arrival and transition airspace is limited mainly by the workload of monitoring and controlling separation \cite{erzberger2004transforming}, and that workload rises with traffic complexity in ways that an aircraft count alone does not capture \cite{hilburn2004cognitive,loft2007modeling}. Radar vectoring, the usual merge technique, requires rapid controller decisions and time-critical crew execution, and it produces workload peaks and frequency congestion \cite{boursier2007merging}. Staffing has tightened at the same time. The National Airspace System Safety Review Team counted 1{,}002 fewer fully certified controllers in August 2023 than in August 2012 and found that staffing shortfalls, overtime, and aging technology are rendering the current level of safety unsustainable \cite{nas_srt2023}, and 20 of 26 critical facilities were below the FAA staffing threshold in 2023 \cite{oig2023staffing}. Decision support is the established response. The Center-TRACON Automation System introduced trajectory-based sequencing and spacing advisories for terminal controllers \cite{denery1995center,halverson1992systems}. Its descendants, the Terminal Area Precision Scheduling and Spacing system and the Terminal Sequencing and Spacing (TSAS) system, extend time-based metering into the terminal area with scheduled times at constraint points and speed advisories for the controller \cite{swenson2011design,baxley2012nasa,thipphavong2013evaluation}. In human-in-the-loop simulation these tools raised the share of arrivals that completed a performance-based navigation arrival during high demand from 42\% to 92\% and reduced inter-arrival spacing error by 25 to 35\% \cite{robinson2015enabling}. These results indicate that decision support lowers workload while improving efficiency and spacing consistency. The terminal component has not reached the field. The FAA deferred TSAS in 2023 without a new implementation date \cite{gao2023atc}. In these systems, and in the arrival scheduling literature more broadly \cite{bennell2011airport,ng2024optimization}, the output is a time or an advisory, and the trajectory that realizes it remains with the controller and the crew.

A previous study computed that trajectory from the horizontal view \cite{pang2026trajectory}. It expressed the three-segment arrival path as a closed-form function of the base-leg extension and committed, for each aircraft entering the Atlanta Large TRACON (A80), an extension distance and a commanded speed per segment that together satisfy wake separation at the FAF. The output was a trajectory a controller can issue and a pilot can fly. It was not yet a 4D trajectory in the TBO concept. The vertical profile was fixed, the commanded speeds stood in for what the aircraft would fly, and the fuel objective was represented by path-stretch and speed penalties. The descent, the deceleration schedule, and the flap configuration were left to the FMS and a published procedure. A 4D trajectory includes the vertical profile together with the time at each point \cite{faa_tbo_2026,icao2018tbo}, and for an idle-thrust descent the two are inseparable. The time an aircraft takes to descend depends on its drag history, which depends on when each flap detent is deployed and where the glideslope is captured. The same choices set the fuel burn \cite{park2015optimal,park2016vertical}, and they are bounded by certified placard speeds and by stabilized-approach criteria at 1{,}000 ft above ground \cite{ho2007methodology,kendall2020stochastic}. A 4D commitment in terminal airspace thus needs a scheduler that chooses the descent design together with the lateral path, a trajectory model that resolves configuration changes instead of a point-mass travel time, verification of each candidate against stabilized-approach criteria, and a treatment of wind at the aircraft so that the committed descent stays flyable at the wind actually encountered. The minimum descent track length of a design then bounds the lateral extension from below, which couples the two decisions.

The gain from the vertical decision comes from fuel-efficient descent procedures. The continuous descent approach (CDA) holds the engines near idle from top of descent to the runway and removes the level segments of a step-down arrival. Flight tests at Louisville and Los Angeles measured the fuel and noise reductions on revenue traffic \cite{clarke2004continuous,clarke2013optimized}, tailored arrivals uplinked the full profile by datalink \cite{coppenbarger2009field}, and the delayed deceleration approach (DDA) extends the idea by keeping the aircraft clean and fast deep into the arrival \cite{reynolds2016delayed,thomas2021modeling,dumont2012fuel}. These benefits hold for isolated aircraft. In dense traffic the spacing buffer that protects an uninterrupted descent consumes runway capacity, so the procedures are used mainly at night or in light traffic \cite{robinson2010benefits,cao2011evaluation,gui2023optimal}. The two problems meet where delay is absorbed. A lateral-only scheduler buys time with level track miles. A descent that configures earlier or captures the glideslope farther from the runway also buys time, often at lower fuel cost, and a scheduler that treats the descent as a fixed edge travel time cannot make that trade. Today the flight crew closes the gap. A pilot receives a vector, judges whether the aircraft can comply, and is required to tell the controller when an instruction would produce an unstabilized approach \cite{faa_ac91_79b}. That arbitration works, but it makes terminal trajectories hard to predict and keeps fuel-efficient descents out of dense traffic. Reduced-crew and autonomous operations remove the arbitrator. NASA and EASA have studied single-pilot and extended minimum-crew airline operations \cite{comerford2013single,bilimoria2014conceptual,schmid2020progressing,easa2022emco}, and NASA is funding research on autonomous cargo operations at scale \cite{nasa2021uli}. In that setting the ground plan and the airborne plan must agree before the descent begins, and the agreement must name the configuration changes as well as the times. The central question of this study is \textit{how much fuel can terminal arrival management save when the vertical descent design and the lateral path extension of every aircraft are chosen together, and how much of that saving survives the published procedure structure that the operation actually flies?}

Answering it requires decision variables that current scheduling formulations do not include. This study treats flap-deployment trigger speeds and glideslope-capture distance as vertical decision variables, bounded by the certified placard windows of each aircraft type. Each candidate is evaluated in a six-degree-of-freedom (6DOF) trajectory simulation. The simulation returns descent time, fuel burn, minimum descent track length, and a stabilized-approach verdict, so stabilization becomes a hard constraint on a simulated trajectory. To keep the fleet scheduler tractable, the design lattice is precomputed once per aircraft type, architecture, and wind node. The online scheduler then uses table lookups.

This paper extends the lateral scheduler of \cite{pang2026trajectory}. The vertical decision set builds on a companion study that optimized the flap trigger speeds and the glideslope-capture distance of a single-aircraft descent under wind uncertainty and verified each candidate design in physics-based simulation \cite{pang2026optimal}. Each arrival still extends its base leg by a variable distance upstream of the FAF, which preserves direct control over arrival time, and the commanded segment speeds are replaced by the speed histories of idle-thrust descent designs. The framework selects the landing order, track miles, glideslope-capture distance, and flap trigger schedule in one decision, and the committed object is a verified 4D trajectory rather than a target time. Two A80 case studies instantiate the framework. The first uses a free descent path to measure the value of full 4D freedom. The second applies the same traffic, winds, and fleet to the six published Runway 8L arrival flows. Repeating that experiment with charted final-approach altitude floors enforced and relaxed prices the published structure directly.

The main contributions of this work are fourfold.

\begin{itemize}
\item We propose a unified optimization formulation that commits the landing order, the track miles, the glideslope-capture distance, and the flap trigger schedule of every arrival, with a fuel objective in place of the surrogate throughput, path-stretch, and speed penalties of the lateral framework and with safety, fuel, and delay resolved in strict priority order. The aircraft configuration thereby becomes a fleet-level decision variable, optimized inside certified placard windows, whereas earlier work tuned the flap schedule for a single aircraft \cite{ho2007methodology,hogenhuis2008optimization,kendall2020stochastic,pang2026optimal} or fixed it in a point-mass descent \cite{park2015optimal,park2016vertical}. A wind-aware backward plan and a 6DOF forward simulation \cite{nagle2009tasat,ren2007separation} verify each candidate, so the stabilized-approach criteria act as hard constraints and the simulated minimum descent track length couples the descent to the lateral extension.

\item Wind uncertainty is modeled with an explicit information structure. Each aircraft observes its along-track wind at the metering gate and re-anchors its descent plan, yielding a two-stage stochastic program with complete recourse whose solution factorizes into per-node problems \cite{park2016vertical}. A crossed Monte Carlo design separates single-aircraft wind effects from fleet-level variance, and shows that wind moves individual descent fuel by 34--81\% yet explains at most 4\% of fleet fuel variance.

\item The CDA and DDA are compared on a common $3.0^\circ$ final and a shared capture-distance grid, which removes the glide-path-angle effect that earlier DDA demonstrations left mixed with the deceleration schedule \cite{reynolds2016delayed,thomas2021modeling,turgut2019effects}. At the matched angle the DDA saves 8--18\% relative to an already optimized CDA, and the two architectures produce ground-noise footprints within a few percent of each other \cite{eurocontrol_anp}.

\item The framework is evaluated on the six published Runway 8L arrival flows at Atlanta in east-flow operations \cite{ZTL_A80_LOA_2022,faa_dtpp_katl}, sweeping five sequencing policies and ten demand levels with 50 seeds each and common random numbers. The asymmetric flow geometry makes delay authority unequal across flows, and running the same traffic with the charted altitude floors enforced and relaxed shows that the floors remove 31\% of the design lattice and cost 3.6--6.7\% of fleet fuel.
\end{itemize}

In the rest of the paper, \Cref{sec:related_work} reviews arrival scheduling, terminal airspace optimization, and fuel-efficient descent procedures. \Cref{sec:methodology} presents the co-optimization framework. \Cref{sec:case1} reports the free-descent case study, and \Cref{sec:case2} reports the published-flow case study. \Cref{sec:discussion} discusses implications and limitations. \Cref{sec:conclusion} concludes.

\section{Literature Review}\label{sec:related_work}

The problem this paper addresses sits between two research communities that have grown largely in parallel. One community optimizes when aircraft land and in what order, with the trajectory reduced to a travel time. The other designs the descent of an individual aircraft and takes the schedule as given. We review each in turn, then the small number of studies that couple them, and then state the gap.

\subsection{Runway and Arrival Scheduling}\label{subsec:rw_runway}

Runway scheduling research is usually grouped into landing scheduling, departure scheduling and runway sequencing, with solution methods that span exact dynamic programming, mixed-integer programming, branch and bound, and a large body of heuristics \cite{bennell2011airport}. The recurring tension in the arrival case is between throughput and fairness. First-come-first-served is simple to explain and easy to accept, but it wastes capacity when a Heavy leader imposes a long wake gap on a Large follower. Constrained position shifting (CPS), proposed by \citet{dear1976dynamic}, bounds how far any aircraft may move from its first-come-first-served rank and recovers much of the optimal throughput while remaining defensible to operators. \citet{psaraftis1980dynamic} gave a dynamic program for the underlying sequencing problem, and \citet{balakrishnan2006scheduling} developed a network-based dynamic program that scales linearly in the number of aircraft and handles pair-specific wake minima, arrival-time windows and same-route precedence explicitly. The mixed-integer formulation of \citet{beasley2000scheduling} established the disjunctive representation that most later work reuses. Extensions cover holding patterns \cite{artiouchine2008runway}, departures and runway crossings \cite{anagnostakis2003runway}, arrival and departure interdependence \cite{diffenderfer2013automated}, and parallel-runway merges \cite{liang2017integrated}. A recent stream sharpens the time inputs to these models with machine learning, either by training the arrival-time predictor against the scheduling cost \cite{lui2024gradient} or by coupling learned predictions to the landing-scheduling program under uncertainty \cite{pang2024machine}. Across this literature the decision model stays a time-based program, and its output is a set of target landing times and delay assignments rather than the trajectories that would realize them.

Terminal-area automation in the United States grew out of the Center-TRACON Automation System, which combined trajectory prediction with the Traffic Management Advisor, the Descent Advisor and the Final Approach Spacing Tool to advise controllers on runway assignment, landing order and terminal spacing \cite{denery1995center,halverson1992systems}. The Final Approach Spacing Tool already produced speed and heading advisories that map onto terminal control actions, which makes it the operational ancestor of the present work. The Terminal Area Precision Scheduling and Spacing system and the Terminal Sequencing and Spacing system carried the same architecture into performance-based navigation arrivals, adding terminal metering to scheduled times at constraint points, controller-managed spacing tools, and flight-deck interval management \cite{swenson2011design,baxley2012nasa,thipphavong2013evaluation,robinson2015enabling}. In all of these systems the trajectory engine serves prediction, the advisories are generated one aircraft at a time, and the vertical profile is whatever the published procedure and the flight management system produce. A parallel line of work reduces the mismatch between the ground prediction and the airborne trajectory through intent exchange and downlink of the extended projected profile \cite{konyak2009improving,torres2011trajectory,bronsvoort2016role,mondoloni2020aircraft}. That work synchronizes two predictions of a trajectory decided elsewhere. The present work decides the trajectory in one place, so that a single object exists for both sides to hold.

\subsection{Terminal Airspace Scheduling and Delay Absorption}\label{subsec:rw_tasp}

A second group of studies pushes the optimization boundary upstream from the runway into the terminal area and formulates the terminal airspace scheduling problem over a network of waypoints and route segments \cite{capozzi2009optimal,desai2016optimization,ng2024optimization}. Each arrival route becomes a chain of nodes joined by edges whose travel-time envelopes come from speed limits or from a set of precomputed profiles, and the optimizer selects time-over-node variables subject to those envelopes and to separation at the runway. Metaheuristics have been applied to the same structure with automated vectoring and point merges as the delay-absorption actions \cite{dhief2023meta,gui2025metaheuristic}, and recent work in this line learns structured vector patterns from historical trajectories \cite{dhief2025automating} or replans the arrival stream through a runway direction change \cite{jiang2025optimized}. The abstraction scales to large terminal networks and accommodates holding stacks and merge points, but it does not compute the trajectory. A time-over-node assignment carries no guarantee that a given aircraft type can fly it from a given entry heading, and the geometric coupling between path stretch and arrival time disappears, since a base-leg extension is a continuous control whose effect depends on the aircraft position and the turn radius. Our earlier study made the opposite abstraction choice and gave up network generality for an analytic path geometry that exposes the base-leg extension and the segment speeds as decision variables \cite{pang2026trajectory}. The present work keeps that geometry and replaces the commanded speeds with the vertical descent design.

The fuel price of delay absorption has also been studied on its own. \citet{nikoleris2016comparison} compared absorbing a required delay by slowing the cruise and the descent against flying extra level distance, for ten Boeing and Airbus types, and found that reducing the descent speed first and the cruise speed second is the cheapest strategy for every type. \citet{saez2023time} compared powered descents along the published route against idle descents along re-negotiated routes and reported that the route freedom, rather than the thrust setting alone, carries most of the difference. \citet{dalmau2018fast} showed that descent trajectories can be updated quickly enough to meet a time constraint in flight. Read from the scheduling side, these results show that the location of delay absorption matters as much as the amount of delay absorbed.

\subsection{Design and Optimization of Fuel-Efficient Descent Procedures}\label{subsec:rw_cda}

The single-aircraft benefit of the continuous descent is well established, and the literature that established it progressed from methodology to demonstration to benefit accounting. The methodology coupled trajectory simulation to impact assessment, beginning with the noise-abatement design framework of \citet{clarke1997systems} and maturing into the fast-time Monte Carlo simulator Tool for the Analysis of Separation and Throughput (TASAT) \cite{ren2007separation,nagle2009tasat}, whose separation-analysis method accounts for inter-flight wind variation and flight management system behavior and was validated against flight test \cite{ren2007modeling}. Demonstrations then carried the procedure into revenue operations, from the Louisville CDA flight test \cite{clarke2004continuous} to daily optimized profile descents at Los Angeles \cite{clarke2013optimized} and tailored arrivals with datalink uplink of the whole profile \cite{coppenbarger2009field}. Benefit studies generalized the result across operating assumptions, including en route speed optimization ahead of the descent \cite{lowther2008enroute}, standard rather than night-only use \cite{cao2011evaluation}, and fleet-wide fuel and emission accounting at Schiphol \cite{ellerbroek2018fuel}. What this line leaves open is whether the procedure parameters themselves are optimal, and what the procedure costs once traffic constrains it.

A smaller body of work treats the procedure parameters as decision variables, and it converges on two findings that shape the present formulation. The first is that the deceleration schedule drives the fuel and emission outcome. \citet{park2015optimal} identified the deceleration schedule, not the descent gradient alone, as the driver in an optimal control setting, \citet{park2016vertical} confirmed the finding under along-track and cross winds modeled as functions of altitude, and \citet{turgut2019effects} isolated the complementary contribution of the flight path angle itself between 1.0 and 4.5 deg, a factor a matched-angle comparison must control. The second is that the flap schedule is both a lever and a risk. \citet{ho2007methodology} optimized the parameters of a noise-abatement approach, including the flap schedule, and showed that variable pilot delay in flap extension makes glideslope interception uncertain and can threaten approach stability, \citet{kendall2020stochastic} therefore treated flap-deployment delay as a random variable and optimized the procedure with a simulation metamodel, and \citet{hogenhuis2008optimization} folded such trade-offs into a multi-objective optimization of area-navigation arrivals. In a companion study, we carried this line to the delayed-deceleration setting and optimized the flap-deployment trigger speeds and the glideslope-capture distance of both architectures to minimize expected fuel under wind uncertainty, subject to a chance constraint on stabilized approach and with each candidate evaluated in 6DOF simulation \cite{pang2026optimal}. This literature is single-aircraft throughout. It designs a procedure for a fleet to fly, but does not include scheduling.

\subsection{Delayed Deceleration Approaches}\label{subsec:rw_dda}

The delayed deceleration approach keeps the aircraft clean and fast far into the arrival and deploys the landing flap late, which lowers both airframe and engine noise over the community and reduces fuel. \citet{dumont2012fuel} quantified the fuel burn reduction potential, \citet{reynolds2016delayed} assessed the noise, and \citet{thomas2021modeling} built a deceleration-rate model for several airframes, reported 4--8 dB noise reductions under the flight track away from the stabilization point, and tested flyability in the 2019 ecoDemonstrator program. Two features of that literature matter here. First, the procedures were flown with a steeper final approach as part of the same design change, so the reported benefit mixes the deceleration schedule with the glide path angle. The companion study \cite{pang2026optimal} separated the two design changes at matched final angles of $3.00^\circ$, $3.50^\circ$, and $3.77^\circ$ for four airframes, found that the glideslope angle carries a large share of the expected-fuel saving, and showed that the demonstrated $3.77^\circ$ final pushes the sink rate on final approach past the 1{,}000~ft/min element of the stabilized-approach criteria. Second, the assessment is per-aircraft. It does not determine whether a fleet of delayed-deceleration arrivals can be sequenced to a runway at capacity or what sequencing costs.

\subsection{Coupling the Descent Procedure to the Schedule}\label{subsec:rw_coupling}

A small number of studies treat the two problems together, and they share one architectural choice: the schedule remains the primary decision, and the descent adapts within the time window the schedule assigns. \citet{robinson2010benefits} established the cost of that arrangement by quantifying how much of the continuous descent benefit survives scheduling constraints in high-density terminal airspace, and found the benefit sensitive to the spacing buffer the procedure requires. Later work reduced that cost from two directions. \citet{saez2020automation} and \citet{saez2021automated} adapted the trajectory to the schedule, generating dynamic arrival routes with mixed-integer programming and then solving an optimal control problem per aircraft to produce a neutral idle descent that meets the assigned time of arrival. \citet{gui2023optimal} adapted the airspace to the procedure instead, offsetting entry paths and extending downwind and base legs so that as many aircraft as possible can fly a continuous descent operation, and \citet{gui2024decision} moved furthest toward a joint decision by generating a menu of candidate descent trajectories per aircraft, covering step-down, continuous descent, and continuous descent with lateral path stretching, and assigning conflict-free trajectories with a variable neighborhood search.

\citet{gui2024decision} is closest in architecture. Their menu contains three procedure categories generated by point-mass trajectory optimization. Here, the menu is a lattice over glideslope-capture distance and flap trigger speeds for two architectures, with each entry verified in 6DOF simulation against stabilized-approach criteria. Their lateral stretch is a route modification priced by length. Here, the base-leg extension is bounded below by the minimum descent track length of the committed design, coupling the two decisions. Their evaluation is deterministic. Here, observed wind enters the information structure, and quadrature expectations separate demand and wind contributions. These studies do not quantify the cost of published crossing restrictions or treat aircraft configuration as a scheduling decision.

\subsection{Research Gap}\label{subsec:rw_gap}

Three gaps follow from this review. First, scheduling and descent-procedure design are usually optimized in separate models. The scheduler therefore cannot use a slower descent design in place of level track miles, even when the descent design is cheaper. The division also separates ground-issued actions from the airborne plan that executes them. That structure reflects current crew-mediated operations, but it is poorly suited to autonomous or reduced-crew arrivals, where the ground schedule and airborne plan must agree in advance and include configuration changes. Second, vertical decisions are often reduced to descent speed, flight path angle, or a choice among named procedure categories. Flap-deployment trigger speeds inside certified placard windows are a real degree of freedom. They set the drag history of an idle-thrust descent and interact directly with stabilized-approach criteria. Optimizing them requires a model that resolves configuration changes and verifies the resulting trajectory, beyond the point-mass abstraction used in most scheduling studies. Third, most evaluations simplify the operating environment. Terminal scheduling formulations often treat wind as a scenario perturbation applied to a fixed plan, so the descent cannot re-anchor top of descent or deceleration segments to the wind actually observed by the aircraft. Fuel-efficient descents are also commonly evaluated in unconstrained or lightly constrained environments with idealized arrival geometry. The resulting fleet-level effect of wind and the fuel cost of published crossing restrictions and asymmetric vectoring geometry are therefore not measured against the same baseline as the benefit.

The formulation addresses these gaps by jointly committing vertical design and lateral extension, including flap triggers in the vertical decision set, verifying candidates in 6DOF simulation, modeling observed wind explicitly, and comparing two case studies that differ only in whether published structure is enforced.

\section{Methodology}\label{sec:methodology}

This section presents the optimization framework for full 4D TBO within the TRACON airspace. It couples the lateral arrival scheduler of \cite{pang2026trajectory} with a vertical descent-procedure design layer adapted from the single-aircraft procedure optimization of \cite{pang2026optimal}, yielding a four-dimensional (4D) trajectory-based operations (TBO) formulation. Each aircraft $i$ extends its base leg by a variable distance $d_i$ upstream of the final approach fix (FAF), which gives the scheduler direct control over arrival time. The main change from the lateral framework is the handling of speed profiles. Different from \citet{pang2026modeling}, commanded segment speeds are no longer decision variables in this formulation. Instead, each arrival's calibrated-airspeed (CAS) history follows from an idle-thrust descent design defined by glideslope-capture distance and flap-deployment trigger schedule, which is a backward-integration process. In such settings, arrival time control comes from the joint choice of vertical design and lateral extension. The fuel cost of each flight is evaluated through flight simulations with real jet engine parameters \cite{ren2007separation, ren2007modeling, ren2009contrast}.

The section proceeds in five parts. \Cref{subsec:flow} defines the mixed-fleet traffic process and pair-specific wake-turbulence separation over four corner feeder gates. \Cref{subsec:geometry} reviews the analytic path geometry of \cite{pang2026trajectory}, which maps extension $d_i$ to total track distance $D_i(d_i)$. \Cref{subsec:wind} introduces the wind model and information structure. \Cref{subsec:vertical,subsec:oracle} define the CDA and DDA design spaces and the wind-aware descent evaluation. \Cref{subsec:sequencing,subsec:cps} present the fleet scheduler, including first-on-final-first-serve (FOFFS) sequencing, per-aircraft joint commitment, and constrained position shifting (FOFFS-CPS$_k$).

Two case studies instantiate the framework. Case Study~1 (\Cref{sec:case1}) considers an idealized future-operations environment with four symmetric corner fixes, a free descent path and no published crossing restrictions, so that aircraft performance, placard limits and the stabilized-approach criteria are the only constraints on the vertical profile. Case Study~2 (\Cref{sec:case2}) transfers the same framework to the six published arrival flows of Runway 8L at Atlanta, where the descent must respect the charted at-or-above crossings on the common final approach course and the arrival geometry is asymmetric. Running Case Study~2 with those crossings enforced and relaxed measures the cost of the published structure against a common baseline.

\subsection{Generation of Arrival Flows}\label{subsec:flow}

Arrival traffic originates from the four corner fixes $\mc K=\{\text{DALAS},\text{LOGEN},\text{HUSKY},\text{TIROE}\}$ that surround the terminal airspace. Each corner generates an independent shifted-Poisson stream \cite{pang2026modeling}, which is consistent with the near-exponential inter-arrival distributions reported at major U.S. airports \cite{willemain2004statistical}. The inter-arrival time at corner $\kappa$ is $\Delta t = T_{\mathrm{sep}} + X_\kappa$ with $X_\kappa \sim \mathrm{Exp}(\lambda_\kappa)$. Here $T_{\mathrm{sep}}=90$~s is the entry separation buffer and $\lambda_\kappa$ is the per-corner arrival rate. We aggregate the streams into the global aircraft set $\mc I = \{1,\dots,n\}$, sorted by entry time $\tau_i$, and each aircraft inherits the entry coordinates $C_i=(X_{C_i},Y_{C_i})$ of its corner.

The descent simulation requires a calibrated performance model for every airframe, so we draw the fleet uniformly from the four simulator-calibrated types in \Cref{tab:ac_params}. Two types are Large (A319, B737-800) and two are Heavy (A340-300, B767-400ER). This fleet replaces the three-class mix of \cite{pang2026trajectory}. We omit the Small class here due to lack of jet engine parameters, and we renormalize the class shares to one half Heavy and one half Large. The minimum time-based separation at the FAF between a leading aircraft of class $c_l$ and a trailing aircraft of class $c_t$ follows the wake-turbulence matrix $T^{\mathrm{wake}}_{\mathrm{sep}}(c_l,c_t)$ of \cite[Table~2]{pang2026trajectory}. The effective required separation for the $k$-th landing slot is,
\begin{equation}
T_{\mathrm{sep}}^{\,k} = \max ~\!\bigl(
T^{\mathrm{wake}}_{\mathrm{sep}}(c_{\Pi(k-1)},c_{\Pi(k)}),\,
T_{\mathrm{rwy}}(\Pi(k))\bigr),
\label{eq:sep_eff}
\end{equation}
where $\Pi$ is the landing sequence and $T_{\mathrm{rwy}}$ is the runway-occupation time.

\begin{table}[t]
\centering
\caption{List of airframes and specifications studied in this work. $V_{\mathrm{ref}}$ is the reference landing speed. $V_{\mathrm{cap}}^{\mathrm{CDA}}$ is the CDA glideslope-capture CAS, with the landing configuration complete before capture. $V_{\mathrm{cap}}^{\mathrm{DDA}}$ is the DDA capture CAS, equal to the landing-flap placard limit so that the landing flap deploys on final.}
\label{tab:ac_params}
\begin{tabular}{lcccccc}
\toprule
Type & Class & Mass (lb) & $V_{\mathrm{ref}}$ (kt) &
$V_{\mathrm{cap}}^{\mathrm{CDA}}$ (kt) &
$V_{\mathrm{cap}}^{\mathrm{DDA}}$ (kt) & $T_{\mathrm{rwy}}$ (s)\\
\midrule
A319      & Large & 130{,}000 & 125 & 165 & 177 & 66\\
B737-800  & Large & 146{,}000 & 141 & 170 & 185 & 62\\
A340-300  & Heavy & 400{,}000 & 122 & 165 & 180 & 85\\
B767-400ER& Heavy & 320{,}000 & 147 & 150 & 183 & 85\\
\bottomrule
\end{tabular}
\end{table}

\subsection{Lateral Path Geometry}\label{subsec:geometry}

The lateral path from the terminal control point (TCP) entry to the FAF consists of a tangent leg of length $d_{L,i}(d_i)$, a radius-to-fix (RF) arc of length $r\,\theta_i(d_i)$ with turn radius $r=2.5$~nmi, and a runway-aligned final extension $d_i \in [0, d_{\max}]$ with $d_{\max}=27.5$~nmi. We restate the closed form of \cite[Sec.~2.2]{pang2026trajectory} here, because the co-optimization uses two of its properties that the lateral study did not need, namely the monotonicity of the track distance in $d_i$ (\Cref{lem:monotone_D}) and its role as an upper bound on the descent track length (\Cref{eq:dmin_floor}).

Place the origin at the runway threshold with the $x$ axis along the extended runway centerline, and let $p_i = (x_i, y_i)$ be the TCP entry point and $F = (x_F, 0)$ the FAF. Let $s_i = +1$ when the aircraft turns onto final from the north and $s_i = -1$ from the south. Extending the final by $d_i$ displaces the RF turn center to,
\begin{equation}
c_i(d_i) = \bigl(x_F - d_i,\; s_i r\bigr),
\label{eq:turncenter}
\end{equation}
so the arc is tangent to the centerline at $(x_F - d_i, 0)$ and the aircraft flies the remaining $d_i$ straight to the FAF. Write $v_i(d_i) = p_i - c_i(d_i)$ for the center-to-entry vector, $\rho_i = \lVert v_i \rVert$, and $R(x,y) = (-y, x)$ for the rotation by a quarter turn. The tangent leg is the exterior tangent from $p_i$ to the turn circle, so its length follows from the right triangle formed by $p_i$, $c_i$ and the tangent point,
\begin{equation}
d_{L,i}(d_i) = \sqrt{\rho_i^2 - r^2}, \qquad \rho_i > r .
\label{eq:tangentleg}
\end{equation}
Projecting $v_i$ onto the radius direction at the tangent point and onto its perpendicular gives the inbound radius vector in closed form,
\begin{equation}
u_i(d_i) = \frac{r^2}{\rho_i^2}\, v_i
         + s_i \frac{r\, d_{L,i}}{\rho_i^2}\, R\,v_i ,
\label{eq:radiusvec}
\end{equation}
which satisfies $\lVert u_i \rVert = r$ by construction, since $\lVert a v + b R v \rVert = \rho_i \sqrt{a^2+b^2}$ and $r^4 + r^2 d_{L,i}^2 = r^2 \rho_i^2$. The arc ends at the centerline tangency point, where the radius vector is $-s_i r \hat y$, so the swept angle is the angle between the two radius vectors,
\begin{equation}
\theta_i(d_i) = \operatorname{atan2}
   \bigl(\lvert u_{i,x} \rvert,\; -s_i\, u_{i,y}\bigr)
   \;\in\; [0, \pi],
\label{eq:arcangle}
\end{equation}
and the total track distance from TCP entry to the FAF is,
\begin{equation}
D_i(d_i) = d_{L,i}(d_i) + r\,\theta_i(d_i) + d_i .
\label{eq:Dtot}
\end{equation}
\Cref{eq:turncenter}--\Cref{eq:Dtot} give $D_i$ as an explicit smooth function of the single scalar $d_i$, with the entry coordinates and the turn direction as parameters. The absolute value in \Cref{eq:arcangle} selects the arc the aircraft actually flies for either entry side.

Two properties of this map matter downstream. The first is monotonicity, which \Cref{subsec:properties} uses to make the per-aircraft commitment exactly solvable. \Cref{lem:monotone_D}, stated in \ref{app:lemmas}, shows that $D_i$ is continuously differentiable and strictly increasing on $[0, d_{\max}]$, with derivative between $0.28$ and $1.98$ for every flow geometry of the two case studies. The lower bound makes the map invertible, and the upper bound of about two makes the extension an expensive control. An aircraft entering with a large arc angle spends up to two track miles for each mile of extension, because the extension lengthens the tangent leg as well as the final segment. \Cref{subsec:case2_flows} shows that this factor, which differs by more than a factor of two across the six published flows, makes delay authority unequal between them.

The second property is that no segment speed appears in \Cref{eq:turncenter}--\Cref{eq:Dtot}. The earlier lateral framework \cite{pang2026trajectory} commanded a speed on each segment. Here, the vertical descent design of \Cref{subsec:vertical} generates the along-track speed history, and $D_i(d_i)$ enters only as available track distance.

\subsection{Wind Uncertainty Modeling}\label{subsec:wind}

We draw a scalar along-track wind component $w$ at the 10{,}000 ft anchor altitude from a truncated normal climatology,
\begin{equation}
w \sim \mathrm{TN}~\!\bigl(0,\sigma_w^2;\,[-\bar w, \bar w]\bigr),
\label{eq:wind_tn}
\end{equation}
with $\sigma_w = 10$~kt in this work. The component materializes over altitude through the power-law profile,
\begin{equation}
W(h) = w\left(\frac{h - h_{\mathrm{rwy}}}
{h_{\mathrm{ref}} - h_{\mathrm{rwy}}}\right)^{1/7},
\label{eq:wind_profile}
\end{equation}
which is anchored at $h_{\mathrm{ref}}=10{,}000$~ft and vanishes at the runway elevation $h_{\mathrm{rwy}}$. We discretize the climatology on $N_w$ evenly spaced nodes $\{w_n\}$ with probability-proportional weights $\{\omega_n\}$, which gives a deterministic quadrature with common random numbers, and each aircraft draws its node independently. For the lateral ground speed, we project the scenario wind vector onto the tangent-leg heading as in \cite[Sec.~2.3]{pang2026trajectory} and \cite{ren2007modeling}. The projection yields the signed along-leg component $w_{L,i}$ and the ground speed $V^{\mathrm{gs}}_i = V^{\mathrm{TAS}}_{\gamma_i} + w_{L,i}$ for level-cruise segments.

Each aircraft observes its wind node $w_i$ at TRACON entry, before the scheduler commits its trajectory. The flight management system then plans the descent with the observed wind entered (\Cref{subsec:oracle}). Conditional on $w_i$, the commitment problem is deterministic. The formulation is therefore a two-stage stochastic program with complete recourse, and its exact solution factorizes into per-node deterministic problems. The reported fleet metrics are quadrature expectations over the climatology \Cref{eq:wind_tn}. This setting is the wind-observable limit of the static procedure-design problem, in which a single descent design is committed against the climatology and stabilized approach is guaranteed only in probability. In the free-descent environment of Case Study~1, no published crossing restriction pins the profile. The wind-aware plan re-anchors top of descent and deceleration segments for the observed wind, and every design in the lattice met stabilization at every wind-envelope node (\Cref{subsec:case1_menus}). The chance constraint regains probabilistic content in the constrained environment of Case Study~2.

\subsection{Descent Architectures and the Vertical Design Space}\label{subsec:vertical}

Each arrival flies one of two idle-thrust descent architectures $a \in \{\mathrm{CDA}, \mathrm{DDA}\}$ from the common entry gate (10{,}000~ft, 240~KCAS, clean) to the runway:
\begin{itemize}[leftmargin=1.4em,itemsep=2pt,topsep=2pt]
\item CDA \cite{clarke2004continuous,clarke2013optimized}: a standard $\gamma_{\mathrm{GS}} = 3.0^\circ$ final glideslope with the landing configuration completed before glideslope capture. The aircraft absorbs the flap ladder on shallow minimum-rate-of-descent deceleration segments during the descent \cite{park2015optimal}.
\item DDA \cite{thomas2021modeling,reynolds2016delayed}: the delayed-deceleration architecture on the same $\gamma_{\mathrm{GS}} = 3.0^\circ$ final, in the continuous-descent parameterization of \cite{pang2026optimal}. The aircraft remains clean at the 240 kt hold CAS deep into the arrival, captures the final at the landing-flap placard speed $V_{\mathrm{cap}}^{\mathrm{DDA}}$, and deploys the landing flap on final.
\end{itemize}

Both architectures fly the same operationally standard $3.0^\circ$ final and share the same capture-distance grid, so the deceleration schedule is the only difference between them. A steeper final would mix a glide-path-angle effect into the comparison. The design vector of aircraft $i$ of type $\gamma_i$ is,
\begin{equation}
\chi_i = \bigl(d_{\mathrm{cap},i},\, x_{i,1},\dots,x_{i,G_{\gamma_i}}\bigr)
\in \mc X^{a}_{\gamma_i},
\label{eq:design_vector}
\end{equation}
where $d_{\mathrm{cap},i}$ is the glideslope-capture distance from the threshold, with platform altitude,
\begin{equation}
h_{\mathrm{cap}} = h_{\mathrm{rwy}} + d_{\mathrm{cap}}
\tan\gamma_{\mathrm{GS}},
\label{eq:hcap}
\end{equation}
and $x_{i,g}$ are the flap-deployment trigger offsets. We partition the flap detents $1,\dots,J_\gamma$ into groups $g$ with identical placard windows $(V_{\min,g}, V_{\max,g})$, midpoint $\mu_g$, and half-width $h_g$. The commanded trigger speed of each group is,
\begin{equation}
V_g = \mathrm{round}\bigl(\mathrm{clip}(\mu_g + x_g;\,
V_{\min,g}, V_{\max,g})\bigr), \qquad |x_g| \le h_g,
\label{eq:trigger_map}
\end{equation}
quantized to the 1 kt command grid and capped 2~kt below the clean hold CAS. A monotone min-cascade $v_j = \min(V_{g(j)}, v_{j-1})$ then sets the per-detent speeds. The approach speed $V_{\mathrm{app}} = V_{\mathrm{ref}} + 5$~kt is pinned at the FAF and held to the threshold for both architectures, so the deceleration completes before the FAF and the aircraft flies the final segment at a steady speed.

A committed trajectory must satisfy four stabilized-approach criteria, encoded in the indicator $s(\chi, w)\in\{0,1\}$ evaluated on the simulated trajectory. The landing flap must be set at the 1{,}000 ft AGL gate, the CAS at that gate must not exceed $V_{\mathrm{ref}}+10$~kt, the CAS at the threshold must be at least $V_{\mathrm{ref}}-10$~kt, and the peak longitudinal deceleration must not exceed $0.12\,g$.

\subsection{Simulation-in-the-Loop Vertical Trajectory Optimization}\label{subsec:oracle}

For a candidate $(\gamma, a, \chi, w)$ this step proceeds in two steps. The first step is a wind-aware backward calculation that constructs the idle-thrust vertical plan from the threshold upstream. The final glideslope and the FAF crossing follow from $\gamma_{\mathrm{GS}}$ and \eqref{eq:hcap}. Each deceleration segment integrates the idle-thrust point-mass dynamics,
\begin{equation}
\frac{dV_T}{dt} = \frac{g\,(T_{\mathrm{idle}}-D_{\mathrm{drag}})}
{W_{\mathrm{ac}}} - g\sin\gamma, \qquad
\gamma = \max\!\Bigl(-\arcsin\frac{\dot h_{\min}}{V_T},\,
\gamma_{\min}\Bigr),
\label{eq:backcalc_dyn}
\end{equation}
with the distance kinematics corrected for the observed wind,
\begin{equation}
\frac{dD}{dt} = V_T\cos\gamma - W(h),
\label{eq:backcalc_wind}
\end{equation}
so the plan allocates deceleration distance in a manner consistent with the wind the aircraft will fly in, as a flight management system with entered descent winds would. The clean idle segment then continues to the entry gate. The along-track length of this plan is the minimum descent track distance $\ul D_\gamma^{a}(\chi, w)$, which is a lower bound on the lateral path length,
\begin{equation}
D_i(d_i) \;\ge\; \ul D_{\gamma_i}^{a}(\chi_i, w_i).
\label{eq:dmin_floor}
\end{equation}
The aircraft flies any surplus $\Delta_i = D_i(d_i) - \ul D_{\gamma_i}^{a}(\chi_i,w_i) \ge 0$ level at the entry altitude upstream of the top of descent.

The second step flies the plan forward in the TASAT 6DOF simulation \cite{nagle2009tasat,ren2007separation} under the wind profile \eqref{eq:wind_profile}. The run yields the descent time to the FAF $t^{\mathrm{des}}$, the fuel burn from the entry gate to the threshold $F^{\mathrm{des}}$, and the stabilization indicator $s$. The level-cruise calibration $(K_\gamma, V^{\mathrm{TAS}}_\gamma)$ gives the fuel per unit distance and the true airspeed at the entry gate. The FAF arrival time and the total fuel of aircraft $i$ under design $\chi_i$ and extension $d_i$ then follow in closed form,
\begin{align}
t_i(\chi_i, d_i; w_i) &= \tau_i + t^{\mathrm{des}}_i
+ \frac{3600\,\Delta_i}{V^{\mathrm{gs}}_i},
\label{eq:time_link}\\
F_i(\chi_i, d_i; w_i) &= F^{\mathrm{des}}_i
+ K_{\gamma_i}\,\frac{V^{\mathrm{TAS}}_{\gamma_i}}
{V^{\mathrm{gs}}_i}\,\Delta_i .
\label{eq:fuel_link}
\end{align}
Level flight at the entry altitude is the most expensive way to cover distance. Fuel-optimal commitments therefore absorb track miles inside the descent whenever a design with a larger $\ul D$ and earlier configuration is available. The extension $d_i$ becomes part of the descent path as a consequence of optimality. Every evaluation is cached, so the fleet layer performs only table lookups.

\subsection{Fleet Co-Optimization Problem and Solution}\label{subsec:sequencing}

For a fixed architecture $a$ and the observed arrival set $\mc I$, the terminal-area co-optimization selects the landing order and the per-aircraft trajectory. We state it as the mathematical program,
\begin{subequations}\label{eq:fleet_problem}
\begin{align}
J^\star = \!\!\min_{\Pi,\,\{\chi_i, d_i\},\,\{\sigma_k\}}\ \
  & W_{\mathrm{safe}}\!\sum_{k=2}^{N}\!\sigma_k
    \;+\; W_{\mathrm{fuel}}\!\sum_{i\in\mc I}\!F_i(\chi_i, d_i;\, w_i)
    \;+\; W_{\mathrm{delay}}\!\sum_{i\in\mc I}\! t_i
    \label{eq:fp_obj} \\
\text{s.t.}\ \
  & t_{\Pi(k)} \ge t_{\Pi(k-1)} + T^{\,k}_{\mathrm{sep}} - \sigma_k,
    \quad k = 2,\dots,N, \label{eq:fp_sep} \\
  & s(\chi_i, w_i) = 1, \quad i \in \mc I, \label{eq:fp_stab} \\
  & D_i(d_i) \ge \ul D^{a}_{\gamma_i}(\chi_i, w_i),
    \quad i \in \mc I, \label{eq:fp_floor} \\
  & \chi_i \in \mc X^{a}_{\gamma_i},\ \ d_i \in [0, d_{\max}],
    \quad i \in \mc I, \label{eq:fp_box} \\
  & \sigma_k \ge 0, \quad k = 2,\dots,N, \label{eq:fp_slack} \\
  & \Pi \in \operatorname{Perm}(\mc I), \label{eq:fp_perm}
\end{align}
\end{subequations}
in which $N = |\mc I|$ is the number of arrivals, $t_i \equiv t_i(\chi_i, d_i; w_i)$ is the FAF arrival time \eqref{eq:time_link}, and $T^{\,k}_{\mathrm{sep}}$ is the required separation of \eqref{eq:sep_eff}.

The decision variables are the landing order $\Pi$, a permutation of $\mc I$, the vertical design $\chi_i = (d_{\mathrm{cap},i}, x_{i,1}, \dots)$ of each aircraft, drawn from the finite lattice $\mc X^{a}_{\gamma_i}$ of \Cref{subsec:vertical}, and the base-leg extension $d_i$. As in \cite{pang2026trajectory}, for each landing rank $k = 2,\dots,N$ we introduce a soft separation slack variable $\sigma_k \ge 0$, which permits a separation violation at the cost of a large penalty $W_{\mathrm{safe}}$ in the objective. The architecture $a$ is fixed per experiment (a homogeneous fleet) and is not optimized.

The objective \eqref{eq:fp_obj} follows the weighted form of \cite{pang2026trajectory}. The penalty hierarchy $W_{\mathrm{safe}} \gg W_{\mathrm{fuel}} \gg W_{\mathrm{delay}} > 0$ encodes the operational priorities of safety first, then fuel, followed by a mild earliest-arrival regularizer. Because $W_{\mathrm{safe}}$ is large enough that no attainable fuel saving can pay for a separation shortfall, the optimizer drives the total slack to its minimum before it minimizes fleet fuel. The scheduler of \Cref{alg:commit} resolves this hierarchy exactly, in strict priority order, so we never tune numerical weights. Relative to \cite{pang2026trajectory}, the single fuel term replaces the surrogate throughput, path-stretch, and speed terms, because delay and path stretch reach this problem only through the fuel they cost \eqref{eq:fuel_link}. The delay term that remains is a regularizer. It breaks ties among fuel-equal designs in favor of the earliest arrival, which protects the slots of later aircraft.

The constraints impose soft wake/runway separation at the FAF \eqref{eq:fp_sep}, hard stabilized-approach feasibility at the observed wind node \eqref{eq:fp_stab}, the descent-feasibility floor \eqref{eq:fp_floor}, and the lattice membership and box bounds \eqref{eq:fp_box}. The slack $\sigma_k$ records saturation of available delay authority instead of making the instance infeasible. The floor constraint repeats \eqref{eq:dmin_floor} and couples vertical design to lateral path length. There is no holding variable. All delay is absorbed by descent design and extension, and any residual shortfall appears as $\sigma_k$.

We solve the program by rolling-horizon commitment. Program \eqref{eq:fleet_problem} couples the fleet only through the shared order $\Pi$ in the separation constraints \eqref{eq:fp_sep}. We fix $\Pi$ with a first-come-first-served rule and then commit the per-aircraft variables in a single online forward pass that touches each aircraft once, as in \cite{pang2026trajectory}. FOFFS sorts aircraft by the nominal estimated time of arrival at the FAF, evaluated at zero extension on the type's reference design $\bar\chi_\gamma$ (midpoint triggers at the base capture distance) under the observed wind,
\begin{equation}
\mathrm{ETA}^{\mathrm{nom}}_i = t_i(\bar\chi_{\gamma_i}, d_i=0;\, w_i).
\label{eq:foffs_eta}
\end{equation}
We also evaluate first-entry-first-serve (FEFS) ordering, which sorts aircraft by TCP entry time and is the fairness reference of \cite{pang2026trajectory}. The commitment step is identical under either order.

When aircraft $i$ reaches slot $k$, the committed predecessor fixes the required slot time $t^{\mathrm{req}}_i = t_{\Pi(k-1)} + T^{\,k}_{\mathrm{sep}}$, and aircraft $i$ solves the per-aircraft restriction of \eqref{eq:fleet_problem} under the fixed order,
\begin{equation}
\begin{aligned}
\min_{\chi \in \mc X^{a}_{\gamma_i},\; d \in [0, d_{\max}]}\quad
& F_i(\chi, d;\, w_i)\\
\text{s.t.}\quad
& t_i(\chi, d;\, w_i) \ge t^{\mathrm{req}}_i, \quad s(\chi, w_i) = 1,\\
& D_i(d) \ge \ul D^{a}_{\gamma_i}(\chi, w_i),
\end{aligned}
\label{eq:joint_commit}
\end{equation}
and records the rank-$k$ slack $\sigma_k = \max\{0,\, t^{\mathrm{req}}_i - t_i(\chi, d_{\max})\}$ when even the slowest stabilized design at $d = d_{\max}$ cannot meet $t^{\mathrm{req}}_i$ (we perform no holding at the TRACON boundary, and residual delay belongs to upstream Center metering, which lies outside the present scope). \Cref{subsec:properties} shows that \eqref{eq:joint_commit} is solved to optimality by enumerating the cached lattice and bisecting, for each design, for the smallest feasible extension. The scheduler first uses descent retardation through a slower, earlier-configuring design. It adds level path stretch only when the design menu cannot absorb the required delay, and it breaks fuel ties by earliest arrival to protect later slots. \Cref{alg:commit} summarizes the forward pass.

\begin{algorithm}[t]
\caption{FOFFS forward pass with joint (design, extension) commitment}
\label{alg:commit}
\begin{algorithmic}[1]
\Require aircraft set $\mc I$ with $(\tau_i, C_i, \gamma_i, w_i)$;
cached evaluations $(\ul D, t^{\mathrm{des}}, F^{\mathrm{des}}, s)$;
wake matrix
\Ensure committed $\{\chi_i, d_i, t_i\}_{i\in\mc I}$,
$\{\sigma_k\}_{k=2}^{N}$
\State sort $\mc I$ by $\mathrm{ETA}^{\mathrm{nom}}_i$
\eqref{eq:foffs_eta} \Comment{FOFFS order $\Pi$}
\State $t_{\mathrm{prev}} \gets -\infty$
\For{$k = 1,\dots,N$;\quad $i = \Pi(k)$}
  \State $t^{\mathrm{req}}_i \gets t_{\mathrm{prev}} +
  T^{\,k}_{\mathrm{sep}}$
  \For{each stabilized design $\chi \in \mc X^{a}_{\gamma_i}$ at $w_i$}
    \State $d_{\mathrm{floor}} \gets \min\{d : D_i(d) \ge
    \ul D^{a}_{\gamma_i}(\chi, w_i)\}$ \Comment{bisection}
    \State $d_\chi \gets \min\{d \ge d_{\mathrm{floor}} :
    t_i(\chi,d) \ge t^{\mathrm{req}}_i\}$;\quad
    $\sigma_\chi \gets 0$
    \If{no such $d$}\ $d_\chi \gets d_{\max}$;\
    $\sigma_\chi \gets t^{\mathrm{req}}_i - t_i(\chi, d_{\max})$
    \EndIf
  \EndFor
  \State commit $(\chi_i, d_i) \gets
  \arg\min_\chi \bigl(W_{\mathrm{safe}}\sigma_\chi
  + W_{\mathrm{fuel}} F_i(\chi, d_\chi)
  + W_{\mathrm{delay}}\, t_i(\chi, d_\chi)\bigr)$
  \Comment{objective \eqref{eq:fp_obj}}
  \State $t_{\mathrm{prev}} \gets t_i$
\EndFor
\end{algorithmic}
\end{algorithm}

\subsection{Structural Properties of the Commitment Problem}\label{subsec:properties}

The per-aircraft restriction \eqref{eq:joint_commit} is mixed, with a finite design lattice and a continuous extension. The enumerate-and-bisect step of \Cref{alg:commit} solves it to optimality. This subsection states the properties behind that result, the weight condition that reproduces the operational priority order, and the interpretation of the reported expectations. Throughout, fix an aircraft $i$ of type $\gamma$, an architecture $a$, and an observed wind node $w_i$. The cached evaluation returns $\ul D = \ul D^{a}_\gamma(\chi, w_i)$, $t^{\mathrm{des}} = t^{\mathrm{des}}(\chi, w_i)$, $F^{\mathrm{des}} = F^{\mathrm{des}}(\chi, w_i)$ and $s(\chi, w_i)$ for each design $\chi$ in the lattice $\mc X^{a}_\gamma$.

\Cref{lem:monotone_tf}, stated and proved in \ref{app:lemmas}, transfers the monotonicity of \Cref{lem:monotone_D} to the decision quantities. For a fixed design $\chi$, the FAF arrival time \eqref{eq:time_link} and the total fuel \eqref{eq:fuel_link} are continuous and strictly increasing in the extension $d$ on the feasible interval $[d_{\mathrm{floor}}(\chi), d_{\max}]$, and their derivatives satisfy $\partial F_i/\partial d = (K_\gamma V^{\mathrm{TAS}}_\gamma/3600)\,\partial t_i/\partial d$. The proportionality in \Cref{lem:monotone_tf} is the analytic form of the statement that level flight at the entry altitude is the only way the extension enters the objective. The extension buys time at a fixed fuel price per second, whereas switching to a design with a larger $\ul D$ buys time at whatever price the descent evaluation reports. The commitment compares those two prices, and it has a closed form.

\Cref{prop:exact}, stated and proved in \ref{app:props}, turns that comparison into an exact algorithm. For the stabilized menu $\mc S = \{\chi \in \mc X^{a}_{\gamma} : s(\chi, w_i) = 1\}$, each design $\chi \in \mc S$ admits a unique smallest feasible extension $d_\chi$, computed by bisection on a strictly increasing continuous function, and the optimum of \eqref{eq:joint_commit} is attained at $\arg\min_{\chi \in \mc S} F_i(\chi, d_\chi)$. Enumerating at most $|\mc X^{a}_\gamma|$ designs and bisecting twice per design therefore solves the per-aircraft commitment exactly. The step that could fail is the stabilization filter. If $\mc S = \emptyset$ for some $(\gamma, a, w_i)$, the aircraft has no committable trajectory and \eqref{eq:joint_commit} is infeasible. \Cref{subsec:case1_menus,subsec:case2_designspace} report that $\mc S \ne \emptyset$ at every cell of the wind quadrature in both case studies, so the hard constraint \eqref{eq:fp_stab} never emptied the feasible set and no chance-constrained relaxation was needed.

\Cref{prop:lex}, also in \ref{app:props}, makes the priority hierarchy of \eqref{eq:fp_obj} precise. Whenever the weights satisfy the separation condition \eqref{eq:weightcond}, the minimizer of the weighted per-aircraft objective in \Cref{alg:commit} coincides with the lexicographic minimizer of the triple of slack, fuel, and arrival time. This removes weight tuning from the reported results. \Cref{alg:commit} implements the lexicographic comparison directly, which is equivalent to solving \eqref{eq:fleet_problem} with any weights satisfying \eqref{eq:weightcond}, and the delay term acts only as a tie-break among fuel-equal designs, in favor of the earliest arrival.

Finally, \Cref{prop:ws} in \ref{app:props} fixes what the reported expectations mean. Every decision variable of \eqref{eq:fleet_problem} is chosen after the corresponding wind node is observed, so the reported fleet metrics estimate the wait-and-see value $\mathbb E_{\boldsymbol w}[J^\star(\boldsymbol w)]$ of the stochastic program, where $\boldsymbol w = (w_1,\dots,w_N)$ collects the independently drawn per-aircraft wind nodes and $J^\star(\boldsymbol w)$ is the optimal value of \eqref{eq:fleet_problem} at that realization. The bound \eqref{eq:wsbound} then places this value below the expected cost of any wind-blind design rule. Two consequences follow. First, the savings reported in \Cref{sec:case1,sec:case2} are the favorable limit of the procedure-design problem. A single published design committed against the climatology cannot beat them, so the values are upper bounds for a wind-blind procedure. Second, the per-aircraft descent evaluation depends on $w_i$ alone, not on $\boldsymbol w$. The cache is indexed by $(\gamma, a, \chi, w)$ and has size $|\Gamma|\,|A|\,|\mc X|\,N_w$, rather than exponential size in the fleet. With four airframes, two architectures, a 16 point lattice and 11 wind nodes this is $1{,}408$ simulated trajectories, and every scenario of every demand level in both case studies reads from that single table.

One forward pass of \Cref{alg:commit} costs $O\bigl(N\,|\mc X|\,\log(1/\epsilon)\bigr)$ cache lookups, where $\epsilon = 10^{-4}$~nmi is the bisection tolerance, since each of the $N$ aircraft enumerates at most $|\mc X| = 16$ designs and performs two bisections per design. No trajectory is simulated online, and no aircraft is revisited once committed. The Phase-1 re-sequencing of \Cref{subsec:cps} dominates the cost at $k = 3$, and we state its complexity there.

\subsection{FOFFS with Constrained Position Shifting}\label{subsec:cps}

FOFFS is a heuristic order. With a heterogeneous fleet, re-pairing nearby aircraft can avoid the large wake penalties behind Heavy leaders \cite{balakrishnan2006scheduling}. We also evaluate FOFFS-CPS$_k$, in which the landing rank of every aircraft may deviate from its FOFFS rank by at most $k$ positions, $|r^\star(j) - r^{\mathrm{FOFFS}}(j)| \le k$. We consider $k \in \{1,2,3\}$, the range recommended in the CPS literature \cite{balakrishnan2006scheduling,beasley2000scheduling,dear1976dynamic}. As in \cite{pang2026trajectory} the policy has two phases. Phase 1 chooses the order, and Phase 2 executes \Cref{alg:commit} under that fixed order, so \Cref{prop:exact,prop:lex} apply unchanged to the per-aircraft commitment. For $k = 0$ the policy reduces to FOFFS.

We solve Phase 1 with the window dynamic program of Balakrishnan and Chandran \cite{balakrishnan2006scheduling} in place of the MILP of \cite{pang2026trajectory}. The program composes landing times on the per-aircraft nominal arrival-time envelopes $[E_j, L_j]$ that the cached evaluations give. The lower end $E_j$ comes from the fastest stabilized design at its feasibility floor, and the upper end $L_j$ from the slowest design at $d_{\max}$. Times clamp at $L_j$ and the shortfall accumulates as slack, which mirrors the Phase-2 semantics. Its objective is the Phase-1 restriction $W_{\mathrm{safe}}\sum_j \sigma_j + W_{\mathrm{delay}}\sum_j t_j$ of \eqref{eq:fp_obj}, in which the fuel term is inactive because the trajectory variables are not yet assigned.

The state after filling landing position $p$ is the pair consisting of the occupancy pattern of the $(2k{+}1)$-wide window of FOFFS ranks $\{p-k,\dots,p+k\}$ and the identity of the last-placed aircraft, which the wake matrix needs. The position-shift bound restricts the successors of each state to the $2k+1$ ranks inside the window, so the reachable state count is,
\begin{equation}
|\mc Q| \;=\; \mathcal{O}\bigl(N\, 2^{2k+1}\,(2k+1)\bigr),
\label{eq:cps_states}
\end{equation}
which is linear in the fleet size for fixed $k$ and keeps the program tractable where an unrestricted permutation search is not. The value attached to each state is not a scalar. Because slack and time trade against each other, we carry the Pareto frontier of the triple (accumulated slack, accumulated time, last landing time), and a successor is discarded when an incumbent dominates it in all three components.

Carrying the full frontier is exact but its width grows with the fleet size, so we cap it at the $24$ lexicographically best incomparable items per state, which bounds the pass at $\mathcal{O}\bigl(N\, 2^{2k+1} (2k+1)^2\bigr)$ frontier operations. The cap is active at the fleet sizes studied here, so Phase 1 is a bounded-width approximation rather than an exact solve. On synthetic instances of $30$ to $61$ aircraft the capped program returns the same order as the uncapped program in about three quarters of cases, and when the orders differ the capped one gives up at most $54$~s of total slack. This approximation is conservative for the reported CPS benefit. The reported CPS arms understate the achievable slack reduction, so the gap between FOFFS and FOFFS-CPS$_k$ in \Cref{sec:case1,sec:case2} is a lower bound on what exact re-sequencing would deliver.

\section{Case Study 1: Free-Descent Trajectory-Based Operations}\label{sec:case1}

Case Study~1 evaluates the co-optimization in an idealized future-operations environment. No published transition-leg crossing restrictions apply, so aircraft performance, placard schedules, and the stabilization criteria are the only constraints on the descent. The setting isolates the value of full 4D trajectory freedom, and its results bound the procedure-constrained Case Study~2 from above.

\subsection{Experimental Setup}\label{subsec:case1_setup}

For each type and architecture, the capture-distance grid is $d_{\mathrm{cap}} \in \{10.0, 11.5, 12.48\}$ nmi. The 5,000 ft platform lies at 12.48 nmi on the $3.0^\circ$ glideslope, and the 10.0 nmi capture sits at the edge of the standard ILS glideslope service volume. Both architectures share this grid, because the common $3.0^\circ$ final gives them the same capture geometry. We cross these capture distances with five common normalized trigger offsets $x_g = \alpha h_g$, $\alpha \in \{-1,-0.5,0,0.5,1\}$, and we add the production $V_{\min}{+}10$ trigger rule at the base capture. This rule is current practice, and we use it as the Baseline throughout the paper. We discretize the wind climatology \eqref{eq:wind_tn} on $N_w = 5$ nodes over $\pm 20$~kt. The comparison covers the Baseline, which fixes the CDA with $V_{\min}{+}10$ triggers and allows no design choice, the fixed midpoint designs of both architectures, and the co-optimized CDA-FOFFS and DDA-FOFFS arms, which give \Cref{alg:commit} the full stabilized menu. The two co-optimized families also run under FEFS ordering and with CPS$_{1,2,3}$ re-sequencing of the FOFFS order. The Monte Carlo demand sweep sets $\lambda_\kappa = \lambda$ for all corners with $\lambda \in \{1,\dots,15\}$ aircraft per hour per corner, which yields aircraft gate arrival rates between 4 and 60~AC/hr. Each demand level uses 50 seeds with common random numbers across policies.

\subsection{Vertical Design Menus}\label{subsec:case1_menus}

\Cref{fig:menu} shows the design menus available to the scheduler. Each panel plots the $(t^{\mathrm{des}}, F^{\mathrm{des}})$ points of every stabilized design for one airframe at the zero-wind node. The DDA cloud sits below the CDA cloud for every type, because an aircraft that stays clean and fast to the same capture point avoids most of the high-drag configured descent. Both architectures favor the smallest capture distance in the shared grid, 10~nmi. Every lattice design also passed the stabilization check. With the free path and the wind-aware plan, the feasible set never emptied across the $\pm 25$ kt envelope, which supports the argument of \Cref{subsec:wind} that a stabilized-approach chance constraint would be inactive in this environment.

\Cref{tab:champions} quantifies the design-level gains at the zero-wind node. The optimized CDA saves 17--23\% relative to the Baseline, and the optimized DDA saves 24--37\%. At the matched glideslope the DDA advantage over the optimized CDA is 8\% for the B737-800 and B767-400ER and 18\% for the A319 and A340-300, largest for the low approach-speed types that have the widest clean-to-approach speed range to exploit. \Cref{fig:profiles} shows the realized profiles for a Large type (B737-800) and a Heavy type (A340-300). The optimized CDA begins its flap ladder 10--12~nmi from the threshold. The DDA holds the 240 kt clean configuration all the way to the same close-in capture, then deploys the landing flap on the final segment. Of the two types shown, the A340-300 has the wider speed range between the clean hold and the approach speed, and it gains more from the delayed deceleration.

\begin{table}[t]
\centering
\caption{Fuel-optimal stabilized designs at the zero-wind node (fuel from the 10{,}000 ft entry gate to the threshold on the natural path). Savings are relative to the CDA Baseline ($V_{\min}{+}10$ triggers at the base capture distance).}
\label{tab:champions}
\begin{tabular}{lccccc}
\toprule
& \multicolumn{1}{c}{CDA Baseline}
& \multicolumn{2}{c}{Optimized CDA}
& \multicolumn{2}{c}{Optimized DDA}\\
\cmidrule(lr){3-4}\cmidrule(lr){5-6}
Type & $F$ (kg) & $F$ (kg) & saving & $F$ (kg) & saving\\
\midrule
A319       & 135.1 & 103.6 & 23.3\% & 84.8  & 37.2\%\\
B737-800   & 208.5 & 173.3 & 16.9\% & 159.4 & 23.5\%\\
A340-300   & 579.5 & 455.1 & 21.5\% & 371.3 & 35.9\%\\
B767-400ER & 441.8 & 362.7 & 17.9\% & 332.7 & 24.7\%\\
\bottomrule
\end{tabular}
\end{table}

\begin{figure*}[!t]
\centering
\includegraphics[width=\textwidth]{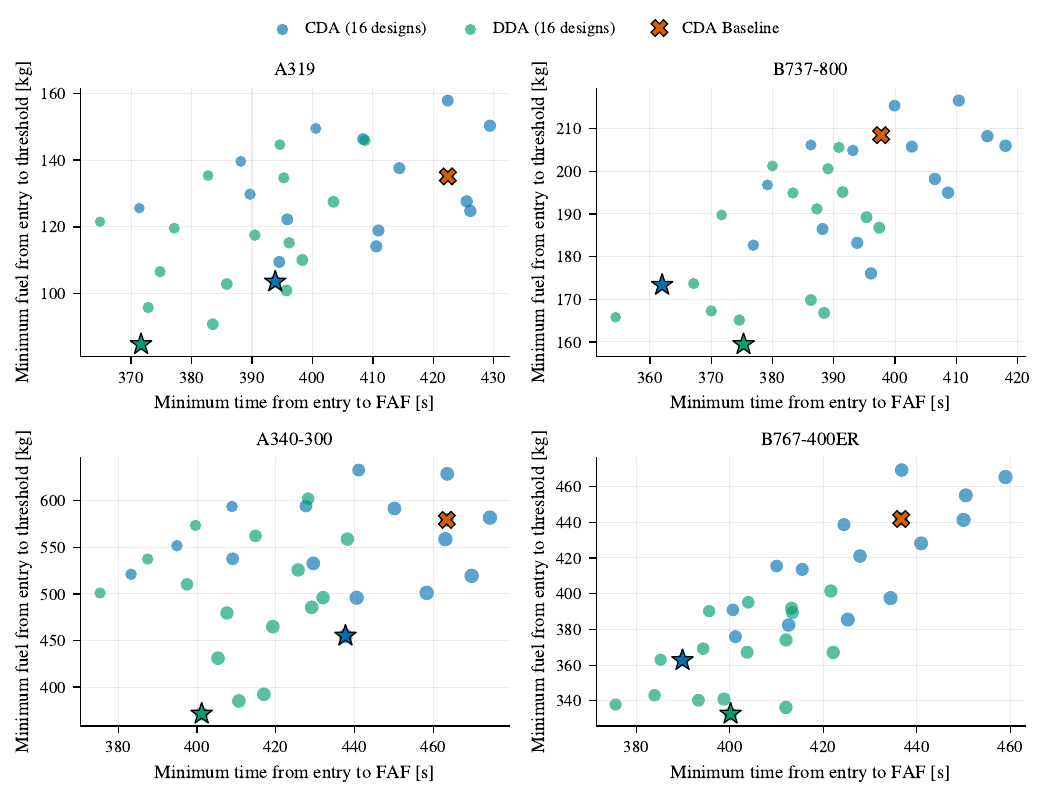}
\caption{Stabilized design menus at the zero-wind node, showing descent time to the FAF versus fuel to the threshold for every lattice design. Marker size scales with the minimum descent track distance $\ul D$, stars mark the fuel-optimal designs, and the cross marks the CDA Baseline. The joint commitment of \Cref{alg:commit} trades along these clouds when absorbing separation delay.}
\label{fig:menu}
\end{figure*}

\begin{figure*}[!t]
\centering
\includegraphics[width=0.9\textwidth]{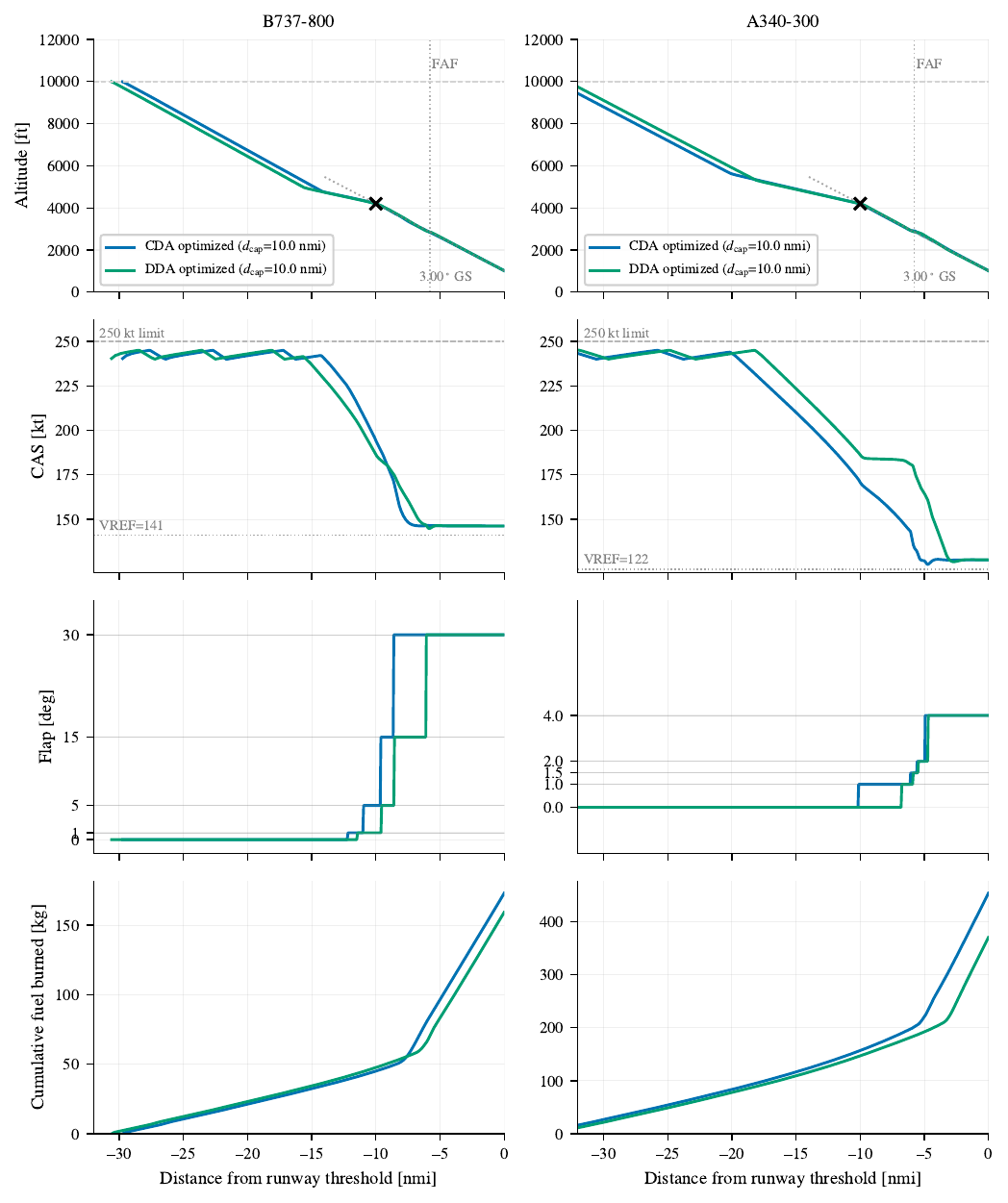}
\caption{Descent profiles at the zero-wind node for a Large type (B737-800, left) and a Heavy type (A340-300, right), simulated on the natural path. The rows show altitude, CAS, flap setting, and cumulative fuel for the fuel-optimal CDA and the fuel-optimal DDA. Dashed lines show the wind-aware plans, and crosses mark the planned glideslope-capture points. On the common $3.0^\circ$ final the DDA stays clean at the 240 kt hold CAS deepest into the arrival, captures at the same 10 nmi distance as the optimized CDA, and deploys the landing flap on the final segment. Flap deployments are shown at the placard detent values.}
\label{fig:profiles}
\end{figure*}

\subsection{Fleet-Level Behavior}\label{subsec:case1_fleet}

\Cref{fig:paths} illustrates committed lateral paths for one uncongested scenario under the Baseline and the co-optimized DDA. \Cref{fig:timespace} shows a congested 60 AC/hr timeline, including wake-limited landing slots, absorbed delay, and delay-authority saturation. Joint commitment differs from a decoupled vertical-then-lateral scheme in two ways. First, the descent-feasibility floor \eqref{eq:dmin_floor} couples the layers. A fast, clean design may require positive extension on short corner geometries, and surplus track miles can make an earlier-configuring design cheaper than level stretch. In these scenarios, 58--65\% of aircraft in the co-optimized families commit their unconstrained fuel-optimal design. The rest trade design against separation slot or surplus distance. Second, the scheduler uses the delay-absorption order of \Cref{subsec:oracle}. It prefers descent retardation, resorts to level stretch when the design menu cannot absorb the delay, and records slack only at saturation.

\begin{figure*}[!t]
\centering
\includegraphics[width=\textwidth]{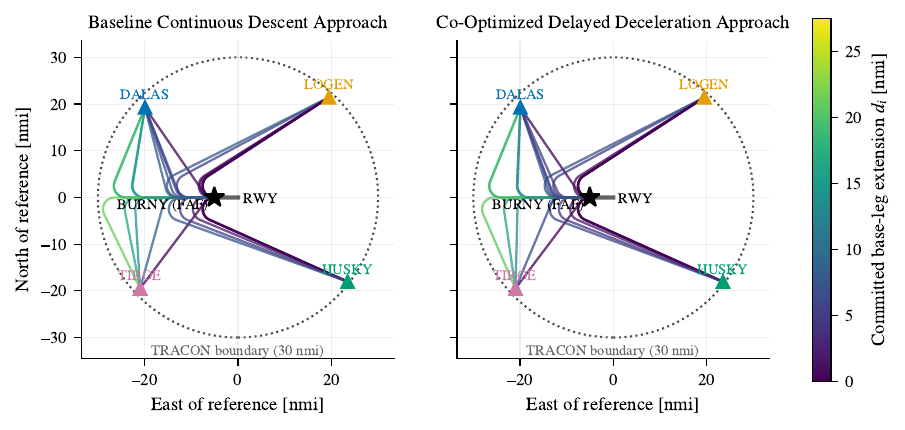}
\caption{Committed lateral paths for one scenario with 27 aircraft arriving from the four corner fixes, under the CDA Baseline (left) and the co-optimized DDA (right). Each path is colored by the committed base-leg extension $d_i$, and the dotted circle marks the 30 nmi TRACON boundary.}
\label{fig:paths}
\end{figure*}

\begin{figure*}[!t]
\centering
\includegraphics[width=\textwidth]{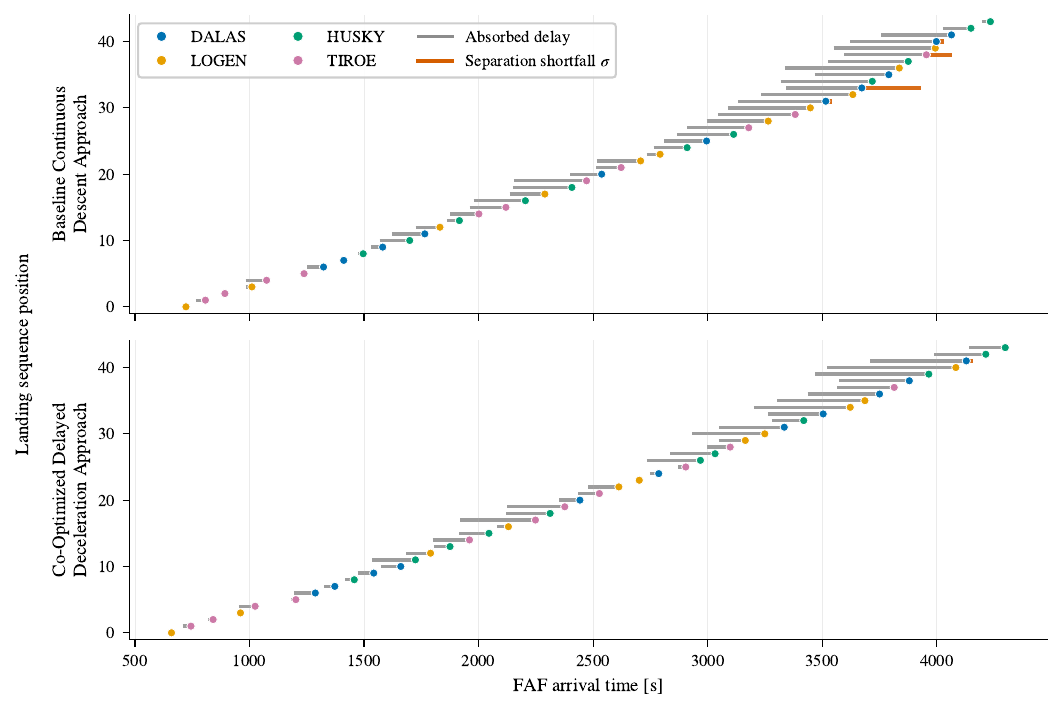}
\caption{FAF landing timeline for a congested 60 AC/hr scenario (44 aircraft, seed 108), comparing current practice (Baseline continuous descent, top) with the co-optimized delayed deceleration under FOFFS-CPS$_2$ (bottom). Each dot is an aircraft at its committed FAF time, colored by its corner fix and stacked by landing sequence position. The gray bar behind each dot is the delay the aircraft absorbs relative to its unconstrained earliest arrival, and a red-orange bar marks any residual separation shortfall $\sigma$. The Baseline saturates the delay authority late in the arrival push and lands four aircraft with shortfalls totaling 392 s, whereas the co-optimized schedule absorbs the same demand with a single 20 s shortfall while burning 22\% less fleet fuel (22{,}300 kg against 28{,}457 kg).}
\label{fig:timespace}
\end{figure*}

\subsection{Effect of Demand on Fuel, Delay, and Saturation}\label{subsec:case1_sweep}

\Cref{fig:sweep_metrics} reports demand sweeps for the CDA and DDA families under FEFS, FOFFS, and CPS$_{1,2,3}$ re-sequencing. The plotted metrics are average fuel per aircraft, average delay, average path stretch $\bar d_i$, and average separation violation versus aircraft gate arrival rate. \Cref{fig:sweep_savings} gives paired fleet-fuel savings relative to the Baseline, and \Cref{tab:sweep} tabulates selected demand levels. The error bars in \Cref{fig:sweep_metrics} pool demand and wind randomness. \Cref{subsec:case1_wind} separates the two contributions.

\begin{table*}[t]
\centering
\caption{Demand-sweep summary (50 seeds per level, common random numbers). The Baseline is current practice (CDA with $V_{\min}{+}10$ triggers). We pair savings per scenario against it and measure both the CDA and DDA families against the same Baseline. $\sigma$ denotes the mean total separation slack per scenario, and $\bar n_{\mathrm{viol}}$ denotes the mean number of violating aircraft. Column headings give the aircraft gate arrival rate in AC/hr.}
\label{tab:sweep}
\begin{tabular}{clccccccc}
\toprule
& & \multicolumn{3}{c}{Fuel saving over Baseline (\%)}
& \multicolumn{2}{c}{$\sigma$ (s)}
& \multicolumn{2}{c}{$\bar n_{\mathrm{viol}}$}\\
\cmidrule(lr){3-5}\cmidrule(lr){6-7}\cmidrule(lr){8-9}
& Policy & 20 & 40 & 60 & 48 & 60 & 48 & 60\\
\midrule
\multicolumn{2}{l}{Baseline (current practice)}
                &  --- &  --- &  --- & 56.6 & 215.1 & 0.88 & 3.34\\
\midrule
\multirow{5}{*}{CDA}
& FEFS          &  8.9 &  5.3 &  8.1 & 102.8 & 335.4 & 1.88 & 5.18\\
& FOFFS         & 15.3 & 13.8 &  8.7 & 46.0 & 189.7 & 0.76 & 3.00\\
& FOFFS-CPS$_1$ & 15.5 & 15.1 & 13.4 & 20.4 & 102.3 & 0.34 & 1.84\\
& FOFFS-CPS$_2$ & 15.5 & 15.5 & 15.2 &  9.3 &  61.3 & 0.26 & 1.64\\
& FOFFS-CPS$_3$ & 15.4 & 15.5 & 15.1 &  6.6 &  47.9 & 0.26 & 1.36\\
\midrule
\multirow{5}{*}{DDA}
& FEFS          & 16.3 & 12.0 & 12.6 & 101.5 & 335.0 & 1.96 & 5.10\\
& FOFFS         & 22.7 & 20.2 & 13.9 & 45.5 & 181.9 & 0.76 & 2.98\\
& FOFFS-CPS$_1$ & 23.2 & 22.3 & 18.8 & 20.3 & 102.3 & 0.38 & 1.76\\
& FOFFS-CPS$_2$ & 23.1 & 22.6 & 20.8 &  8.3 &  58.9 & 0.20 & 1.44\\
& FOFFS-CPS$_3$ & 23.0 & 22.6 & 20.6 &  6.1 &  44.2 & 0.28 & 1.08\\
\bottomrule
\end{tabular}
\end{table*}

Under plain FOFFS the DDA fleet saves about 23\% at low demand and 13.9\% at 60~AC/hr, while the CDA fleet saving falls from about 15.5\% to 8.7\%. Both families lose ground as demand rises. Congestion forces level path stretch on every policy, and that shared fuel cost dilutes the descent-side savings.

FEFS ordering performs worst on both fuel and slack, which matches the lateral study \cite{pang2026trajectory}. Entry order cannot exploit the geometric asymmetry between the corner streams, so FEFS spends far more path stretch at every demand level and saturates earlier. At 60~AC/hr its slack roughly doubles that of FOFFS, at 335--337~s with about 5.1 violating aircraft per scenario, and its paired fuel saving drops to 5--9\% for the CDA family and 12--16\% for the DDA family.

CPS re-sequencing changes little below the saturation onset near 40~AC/hr, where slack first appears. Above the onset, it holds savings nearly flat by recovering slots that FOFFS spends on pairings behind Heavy leaders. At 60~AC/hr with CPS$_3$, the saving is 15.1\% for the CDA and 20.6\% for the DDA. Re-sequencing matters most for safety. At 60~AC/hr, CPS$_3$ reduces mean slack from 182--190~s to 44--48~s and the mean number of violating aircraft from about 3.0 to 1.1--1.4 per scenario. CPS$_2$ already captures most of this gain, consistent with diminishing returns beyond $k=2$ \cite{balakrishnan2006scheduling,pang2026trajectory}.

The midpoint rule burns about 5\% more fuel than the Baseline trigger rule at every demand level. Late deployment appears to be fuel-optimal in this environment. The CDA savings come from optimizing the triggers inside the placard windows, and an arbitrary change to the triggers would not produce them.

\begin{figure*}[!tp]
\centering
\includegraphics[width=0.95\textwidth]{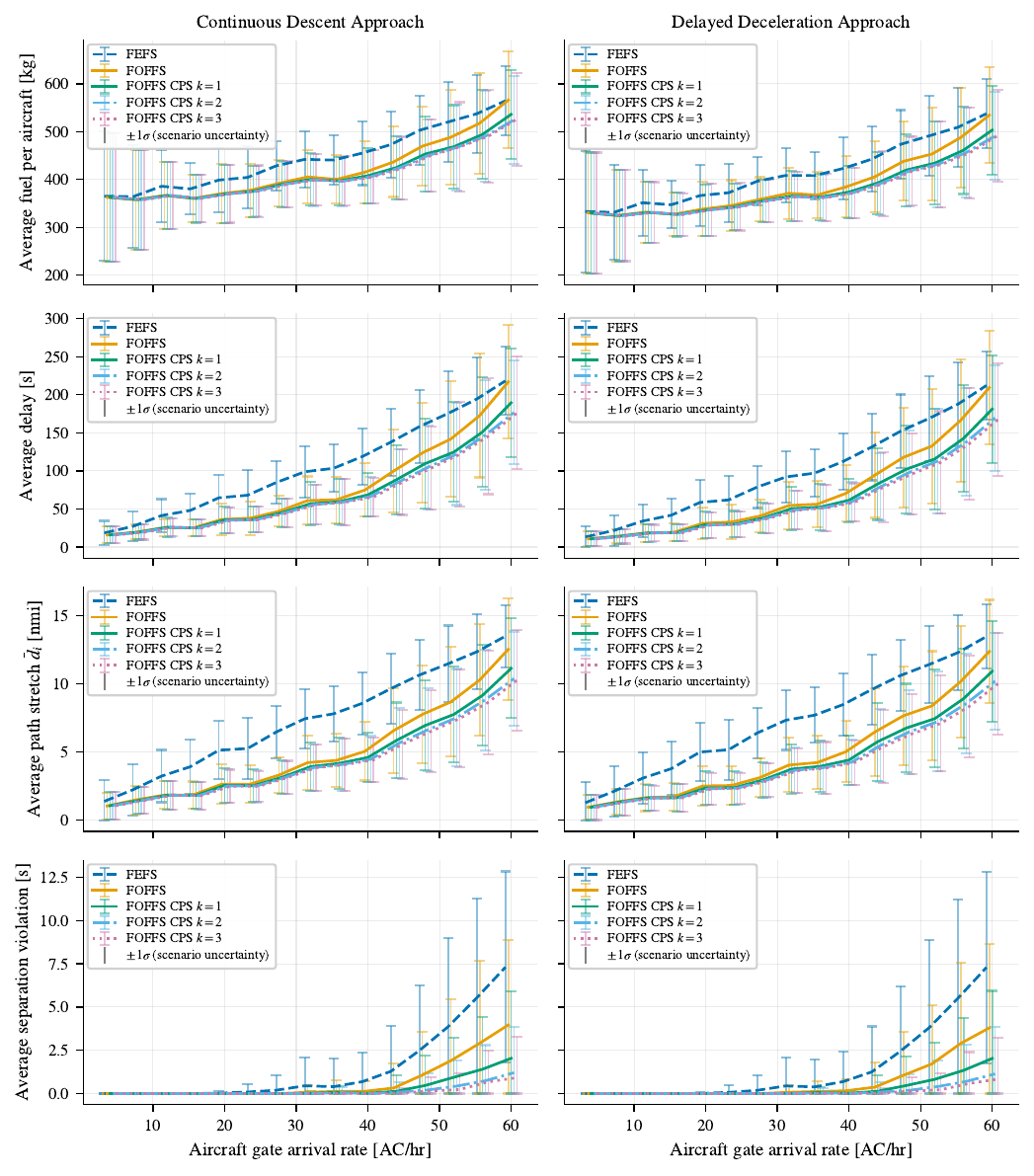}
\caption{Average per-aircraft fuel burn, average delay, average path stretch $\bar d_i$, and average separation violation (rows) versus aircraft gate arrival rate for the CDA family (left column) and the DDA family (right column) under FEFS, FOFFS, and FOFFS-CPS$_{1,2,3}$ (50 seeds per demand level). Within each row the two columns share the vertical axis, so a horizontal read gives the difference between the two architectures directly. Error bars show the $\pm 1\sigma$ scenario (demand and wind) uncertainty, and lower bars are clipped at the physical zero.}
\label{fig:sweep_metrics}
\end{figure*}

\begin{figure*}[!t]
\centering
\includegraphics[width=\textwidth]{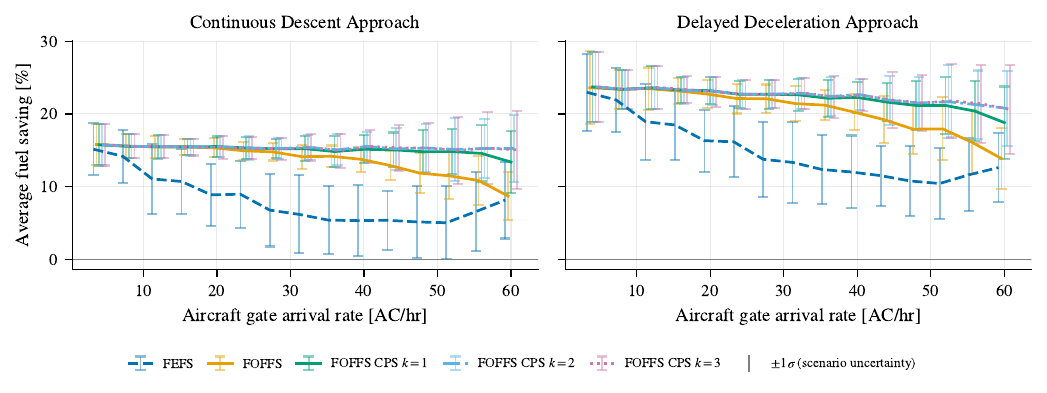}
\caption{Average fleet-fuel saving relative to the CDA Baseline, paired per scenario, versus aircraft gate arrival rate for the CDA family (left) and the DDA family (right) under FEFS, FOFFS, and FOFFS-CPS$_{1,2,3}$ on a shared vertical scale. Error bars show the $\pm 1\sigma$ scenario uncertainty.}
\label{fig:sweep_savings}
\end{figure*}

\subsection{Effect of Wind}\label{subsec:case1_wind}

The metrics of \Cref{fig:sweep_metrics,fig:sweep_savings} are quadrature expectations over the wind climatology, so wind appears in them only through the pooled scenario error bars. This subsection quantifies the wind impact directly, first on the individual descent and then on the fleet metrics.

\Cref{fig:wind_fuel} evaluates the cached design lattice on an 11-node grid spanning $\pm 25$~kt, one node beyond the $\pm 20$ kt climatology envelope. Wind is a first-order effect on the individual descent. Between the 20 kt tailwind node and the 20 kt headwind node the descent fuel of the CDA Baseline rises by 34--81\% depending on airframe (for example 179 to 249~kg for the B737-800 and 450 to 732~kg for the A340-300), because the headwind lowers the ground speed and the wind-aware plan spends longer on the same track. The swing exceeds every policy gap in \Cref{fig:sweep_metrics}. The fuel-optimal design is wind-dependent. The DDA champion re-selects for three of the four airframes across the grid, up to five times for the B767-400ER, so a single wind-blind design would forgo part of the menu benefit. The optimized DDA is also less wind-sensitive than the Baseline. For the A340-300, its tailwind-to-headwind swing is 109~kg, less than half the Baseline swing of 282~kg. The DDA fuel advantage therefore persists across the wind range.

The fleet metrics see a much weaker wind signal. To separate the two randomness sources, we repeat the sweep with a crossed design. Each demand level uses 16 demand realizations and 16 independent redraws of the per-aircraft wind nodes for the FEFS and FOFFS arms of both families. The spread of the scenario mean then splits into a demand component, defined as the standard deviation over demand realizations of the wind-averaged metric, and a wind component, defined as the root-mean-square within-realization standard deviation over wind redraws. The squared components reproduce total variance to within 6\%. For the FOFFS arms, the demand component tracks the total at every demand level, while the wind component stays at 6--18~kg of fuel and 2--10~s of delay per aircraft, which amounts to at most 4\% of the fuel variance and 9\% of the delay variance (the FEFS arms behave the same). The information structure of \Cref{subsec:wind} explains the contrast. Each aircraft draws its wind node independently, so per-aircraft wind effects average out within the fleet. A demand realization, including stream bunching and weight-class mix, shifts the whole scenario. Thus, the error bars of \Cref{fig:sweep_metrics} mainly measure demand uncertainty, while wind acts on each descent individually and is absorbed by wind-aware planning.

\begin{figure*}[!t]
\centering
\includegraphics[width=\textwidth]{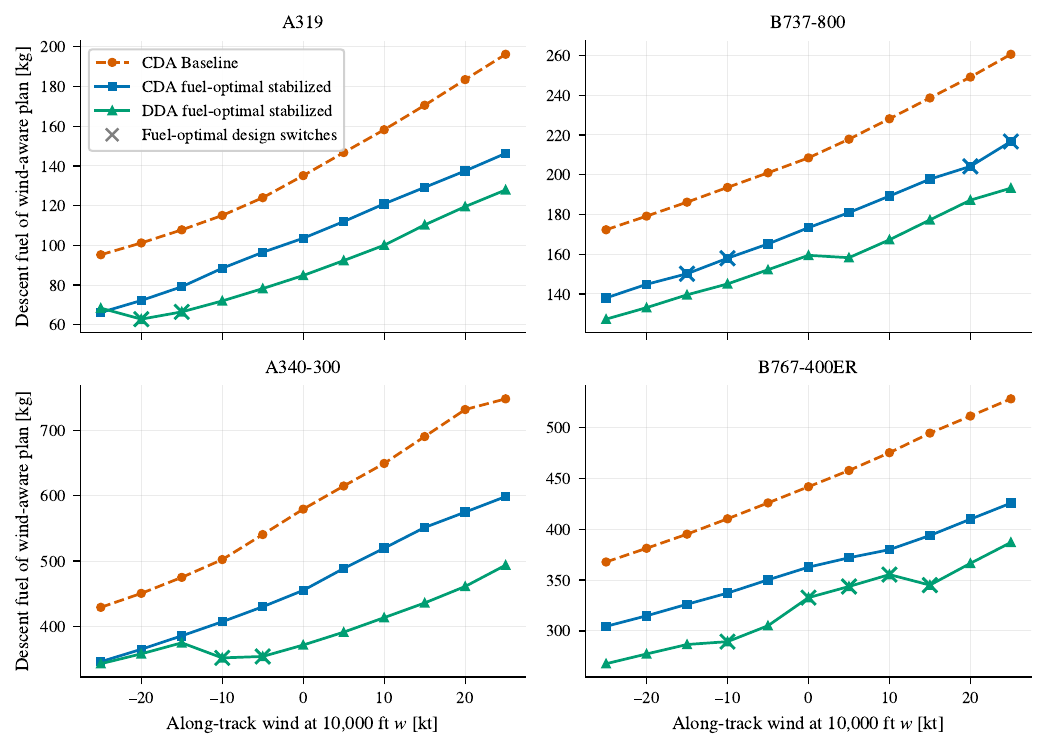}
\caption{Descent fuel of the wind-aware plan versus the along-track wind component $w$ at 10{,}000 ft (headwind positive) for each airframe, comparing the CDA Baseline ($V_{\min}{+}10$ triggers at the 5{,}000 ft platform capture), the fuel-optimal stabilized CDA design, and the fuel-optimal stabilized DDA design on the cached design lattice. Crosses mark the nodes at which the fuel-optimal design changes.}
\label{fig:wind_fuel}
\end{figure*}

\subsection{Noise Footprint}\label{subsec:case1_noise}

The vertical redesign changes where the fleet flies low. We therefore close the case study with the ground-noise consequence. We evaluate the maximum A-weighted sound level $L_{A,\max}$ on a 0.75 nmi ground grid spanning the TRACON. Every 10~s of the replayed scenario (level surplus at 10{,}000~ft, wind-aware descent, final approach), the noise-power-distance curve of each airborne aircraft, taken from the EUROCONTROL Aircraft Noise and Performance (ANP) database \cite{eurocontrol_anp} and interpolated in the logarithm of the slant distance following the Doc~29 convention \cite{ecac_doc29}, gives its level at each grid point. Simultaneous aircraft add on a pressure basis, and the cumulative footprint is the maximum over time. NPD curves of closely related ANP types substitute for airframes without a native entry (A320-family curves for the A319 and B737-800, A330-family curves for the A340-300, and 767-300ER curves for the B767-400ER). The NPD lookup holds the engine state fixed. It therefore does not represent the airframe-noise increment of earlier flap deployment in the DDA, and the comparison isolates the geometric effect of altitude and path.

\Cref{fig:cumnoise} compares the cumulative footprints of the two co-optimized architectures on the congested scenario of \Cref{fig:paths} (seed 8, 27 aircraft). The peak on short final is the same (73.6~dBA) and the 65 dBA area is 10--11~nmi$^2$ for both. On the matched $3.0^\circ$ final, the two 55 dBA footprints are within a few percent, 147~nmi$^2$ for the CDA and 152~nmi$^2$ for the DDA. The difference map is small and shows no systematic sign. Louder streaks along individual feed corridors reflect different committed extensions $d_i$, which shift where each aircraft reaches low altitude, not a systematic altitude difference. Once the steeper final is removed, delayed deceleration no longer produces a quieter descent. The CDA Baseline footprint for the same scenario is the smallest of the three at 144~nmi$^2$, because both fuel-optimal designs capture closer in than the Baseline platform. \Cref{fig:noise_density} quantifies this ordering across demand.

\Cref{fig:noise_density} reports the same metrics across the demand sweep (10 seeds per level). Three observations follow. First, the peak combined $L_{A,\max}$ is flat at 73--74~dBA for every architecture beyond about 12~AC/hr. The overflight of the grid cell closest to the final-approach path sets this level, and in-trail separation keeps simultaneous aircraft too far apart to raise it. The maximum noise level at a point does not depend much on arrival density or on the procedure. What demand changes is the exposed area. Second, the 55 dBA area grows steeply with demand while path stretching pushes the corridors outward, and saturates beyond about 40~AC/hr together with the stretch authority (\Cref{fig:sweep_metrics}). Both co-optimized arms sit slightly above the Baseline on this plateau: at 40~AC/hr the co-optimized CDA is 9--11\% above (173 against 158~nmi$^2$) and the co-optimized DDA 3--5\% above (164~nmi$^2$), because both fuel-optimal families replace the 5{,}000 ft platform with a continuous idle path that captures closer in and crosses 5{,}000~ft near 15~nmi, so each descent runs a few hundred feet lower over the 10--15~nmi band. Fuel optimization therefore carries a small ground-noise penalty, slightly larger for the CDA family, and the DDA attains the fuel saving of \Cref{fig:sweep_savings} at close to the Baseline footprint. Third, the 65 dBA area stays at 9--11~nmi$^2$ for every arm and demand level. It is confined to the short final and the base-turn region that every procedure must fly.

\begin{figure*}[!t]
\centering
\includegraphics[width=\textwidth]{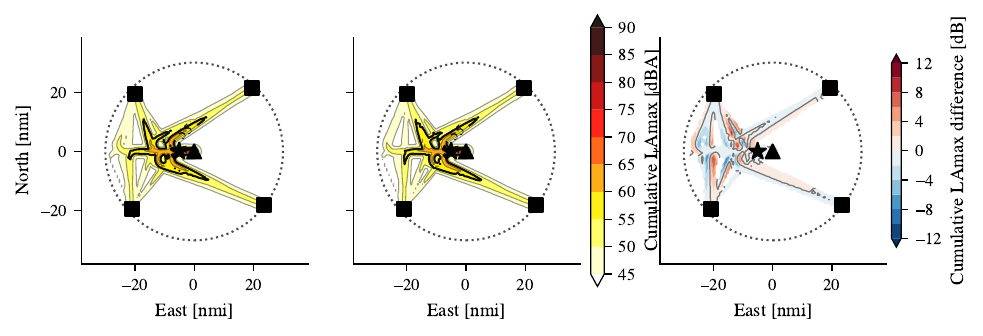}
\caption{Cumulative (time-maximum) combined $L_{A,\max}$ footprint of the congested seed 8 scenario of \Cref{fig:paths}. The left and center panels show the co-optimized CDA and the co-optimized DDA on a shared color scale with the 55 and 65 dBA contours outlined, and the right panel shows the difference between the DDA and the CDA over the exposed area. Levels come from EUROCONTROL ANP noise-power-distance curves with linear pressure summation across simultaneous aircraft on a 0.75 nmi grid. Squares mark the corner fixes, the star marks the FAF, and the triangle marks the runway.}
\label{fig:cumnoise}
\end{figure*}

\begin{figure*}[!t]
\centering
\includegraphics[width=\textwidth]{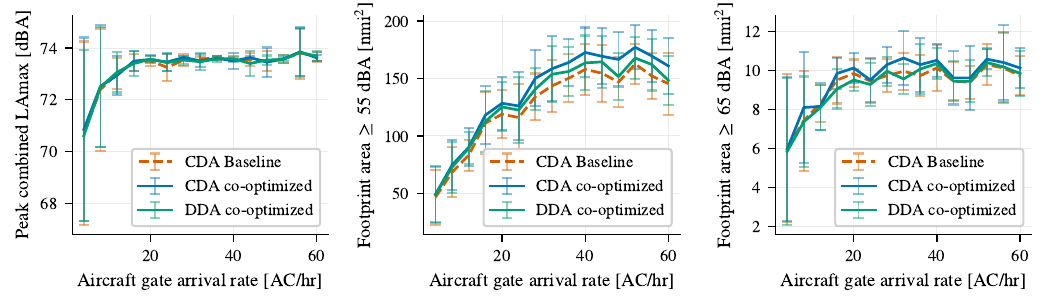}
\caption{Ground-noise metrics versus aircraft gate arrival rate (10 seeds per level, mean $\pm 1\sigma$). The panels show the peak combined $L_{A,\max}$ (left) and the ground area of the cumulative footprint above 55 dBA (center) and 65 dBA (right) for the current-practice CDA, the co-optimized CDA, and the co-optimized DDA under FOFFS.}
\label{fig:noise_density}
\end{figure*}

\section{Case Study 2: Published Arrival Flows to Runway 8L}\label{sec:case2}

Case Study~2 applies the co-optimization to the arrival structure flown in operation. The study asks two questions. First, does the joint commitment of \Cref{sec:methodology} carry over when several published flows enter from different sides of the field and merge onto one final approach course with charted altitude restrictions? The four symmetric corner fixes of Case Study~1 give every aircraft the same delay authority, which is an artifact of the model rather than a property of the airspace. Second, what do the published crossing restrictions cost? Running the same traffic, winds, and geometry with the charted floors enforced and relaxed answers this directly and bounds the benefit reported in Case Study~1.

\subsection{The Published Arrival Structure and the Experimental Setup}\label{subsec:case2_setup}

All experiments use the six published area-navigation arrival flows to Runway~8L at Hartsfield--Jackson Atlanta International Airport. The operation is east-flow traffic inside the same A80 terminal airspace as Case Study~1 \cite{ZTL_A80_LOA_2022,faa_dtpp_katl}. Each flow enters at its standard terminal arrival route boundary fix, flies the charted transition leg to a metering fix, and is then radar-vectored to the final approach course. \Cref{tab:case2_flows} lists the six flows and \Cref{fig:case2_paths} shows them in plan view. The six metering fixes sit between 27.1 and 30.4 nmi from the threshold, which is the same ring the four synthetic corner fixes of Case Study~1 occupy, so metering starts at the same gate condition of 10{,}000 ft, 240 KCAS and clean configuration. The charted legs upstream of the metering fix are context rather than decisions.

\begin{table}[t]
\centering
\caption{The six published Runway 8L arrival flows. The charted leg is the published transition from the boundary fix to the metering fix, and the ring distance is the along-track distance of the metering fix to the threshold. $D_i(0)$ is the track distance from the metering fix to the FAF at zero extension, and the last column is the track distance each flow spends per nautical mile of centerline extension $d_i$.}
\label{tab:case2_flows}
\begin{tabular}{llccccc}
\toprule
\makecell[l]{Metering\\fix} & \makecell[l]{Boundary\\fix}
& \makecell{Charted\\leg (nmi)} & \makecell{Ring\\(nmi)}
& Group & \makecell{$D_i(0)$\\(nmi)}
& \makecell{Track nmi\\per nmi}\\
\midrule
MRCHH & CHPPR & 10.0 & 29.7 & near side & 26.6 & 1.0\\
STHRN & GLAVN & 10.1 & 29.0 & near side & 24.6 & 0.7\\
JNGLE & HALRR & 12.5 & 28.7 & near side & 25.5 & 1.0\\
SMELY & ONDRE & 13.2 & 27.1 & far side  & 35.7 & 1.9\\
HAARY & OZZZI & 10.0 & 30.4 & far side  & 40.4 & 1.9\\
TIZZY & SITTH & 13.0 & 28.4 & far side  & 38.2 & 1.9\\
\bottomrule
\end{tabular}
\end{table}

The vector segment from the metering fix reuses the closed-form geometry of \Cref{subsec:geometry}, re-anchored on the published FAF SCHEL, with a tangent leg, an RF arc of radius $r = 2.5$~nmi, and a centerline extension $d_i$ upstream of SCHEL. The published final-approach fixes lie within 0.02 nmi of the extended Runway 8L centerline, so that centerline is the axis of the local frame. The committed path must stay inside the 30 nmi terminal boundary. Its farthest point lies on the RF turn circle, whose center sits at $(-(|\mathrm{FAF}| + d_i), \pm r)$, so containment requires
\begin{equation}
\sqrt{(|\mathrm{FAF}| + d_i)^2 + r^2} + r \;\le\; R_{\mathrm{TRACON}},
\label{eq:containment}
\end{equation}
which gives $d_{\max} = 21.59$~nmi. Bounding only the intercept point would permit 24.2 nmi and let the arc bulge 2.6 nmi outside the boundary, so we use the tighter form. The cap binds for 9\% of aircraft at the highest demand level and for none below 36 AC/hr. Case Study~1 keeps its own 27.5 nmi authority, which is not containment-limited because its FAF, corner fixes and boundary differ.

The geometry is asymmetric, and that asymmetry is the substantive difference from Case Study~1. The near-side flows reach the FAF in 24.6--26.6 nmi and buy delay at 0.7--1.0 track nmi per nautical mile of extension. The far-side flows start 35.7--40.4 nmi out, because they are vectored across the north or the south side of the field to join the final from the west, and they pay 1.9 track nmi per nautical mile of extension on top of the 8--13 nmi of surplus track miles they carry even at $d_i = 0$.

\begin{figure*}[!t]
\centering
\includegraphics[width=\textwidth]{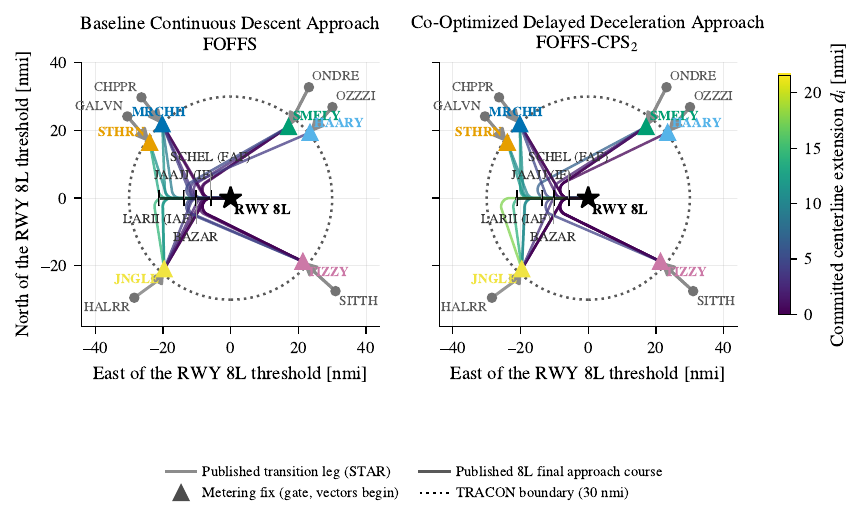}
\caption{Plan view of the six published Runway 8L arrival flows with the committed paths of one 27 aircraft scenario, comparing current practice (Baseline continuous descent under FOFFS, left) with the co-optimized delayed deceleration under FOFFS-CPS$_2$ (right). Gray lines are the charted transition legs, triangles mark the metering fixes at which metering and vectoring begin, and each committed path is colored by its centerline extension $d_i$. The dotted circle is the 30 nmi terminal boundary, which constrains $d_i$ through \eqref{eq:containment}. The two schedules spend almost the same total extension, 117 and 118 nmi, yet the co-optimized fleet burns 20.8\% less fuel.}
\label{fig:case2_paths}
\end{figure*}

The final approach course is common to all six flows and carries the charted at-or-above crossings of \Cref{tab:case2_floors}. These are the binding published restrictions in this environment. We screen them on the wind-aware plan of \Cref{subsec:oracle} at the observed node $w_i$ with a 50 ft tolerance, and we discard any design whose plan crosses a floor low before simulation. Because the plan is FAF-anchored and absorbs surplus track miles level at the gate altitude upstream of top of descent, planned altitude is a function of distance to go alone. The screen is therefore a property of the tuple (type, architecture, design, wind) and does not couple to $d_i$. On the centerline extension, the screen is the charted crossing restriction. Upstream of the intercept, it is the vectoring platform minimum, because aircraft are vectored at or above intercept altitude until glideslope capture. Screening in distance to go is therefore conservative. The SCHEL crossing is met by construction. The plan is on the $3.0^\circ$ glideslope at the FAF and crosses at 2{,}846~ft against the charted 2{,}900~ft, treated here as an at-or-above floor. Flow-specific step-downs upstream of the intercept are not modeled because the vector segment replaces the charted transition downstream of the metering fix.

\begin{table}[t]
\centering
\caption{Charted at-or-above crossing restrictions on the common Runway 8L final approach course. Distances are along-track to the threshold, and the field elevation is 1{,}026 ft.}
\label{tab:case2_floors}
\begin{tabular}{lcl}
\toprule
Fix & Distance (nmi) & Restriction\\
\midrule
LARII (IAF) & 21.1 & At or above 5{,}000 ft\\
JAAJJ (IF)  & 13.7 & At or above 5{,}000 ft\\
BAZAR       & 10.2 & At or above 4{,}000 ft\\
SCHEL (FAF) &  5.8 & 2{,}900 ft, met on the $3.0^\circ$ glideslope\\
\bottomrule
\end{tabular}
\end{table}

The gate condition matches Case Study~1, so the verified simulation-in-the-loop lattice transfers unchanged and the floor screen is the only addition. Each flow generates an independent shifted-Poisson stream with a 90 s minimum gap, the Case-1 convention, over the same four simulator-calibrated airframes of \Cref{tab:ac_params}, and wind uses the same five-node quadrature of the $\pm 20$ kt climatology with each aircraft observing its node at the gate. The demand sweep sets the same per-flow rate on all six flows and spans nominal aggregate rates from 6 to 60 AC/hr, which the 90 s minimum gap delivers as 5.4 to 45.9 aircraft per hour, with 50 seeds per level and common random numbers across policies. The policy arms are the CDA Baseline ($V_{\min}{+}10$ triggers at the 12.48 nmi base capture, which is the charted 5{,}000 ft platform on the $3.0^\circ$ glideslope), the co-optimized CDA and DDA families under FEFS, FOFFS and FOFFS-CPS$_{1,2,3}$, and a floors-relaxed run of each co-optimized family that keeps everything else fixed. Because the six flows merge onto one course and every aircraft can absorb delay on its own centerline extension, bounded position shifting is operationally meaningful here in the same way as in Case Study~1. An aircraft absorbing delay on a longer extension can be overtaken by a trailing aircraft from another flow.

\subsection{Effect of the Charted Floors on the Design Space}\label{subsec:case2_designspace}

The floors remove a specific part of the decision set. Across the four airframes, both architectures, and five quadrature nodes, the screen removes 5 of the 16 lattice designs, or 31\%. The JAAJJ 5{,}000 ft floor accounts for every removal. The pruned cells are exactly the 10.0 nmi glideslope-capture row, whose plans cross JAAJJ between 4{,}700 and 4{,}950~ft. The pruning is systematic and removes the designs preferred by the unconstrained problem. With the floors enforced, the fuel-optimal design is the 11.5 nmi capture for every airframe and both architectures. With the floors relaxed, it is the 10.0 nmi capture. The published structure pins glideslope capture near the charted platform. Two tail cells differ for a separate reason. The A340-300 DDA loses three designs to the stabilization check at the 20 kt tailwind node and two at the 10 kt node, where the wide-body cannot complete its late deceleration. The A319 DDA loses one design at the same node. The stabilized menus of the co-optimized arms never emptied at any quadrature cell, so the stabilization requirement remained enforceable as a hard constraint under the wind-observable information structure of \Cref{subsec:wind}.

\Cref{fig:case2_profiles} shows the resulting trajectories for a Large type and a Heavy type at the zero-wind node, plotted against the charted floor corridor. The three floor-feasible profiles stay above every crossing and differ only in where they place the deceleration. The floors-relaxed CDA optimum is drawn for comparison and crosses JAAJJ about 300 ft low, which is why the screen discards it. \Cref{fig:case2_menu} shows the design menus with the pruned designs marked, and \Cref{tab:case2_champions} quantifies the design-level gains against the same Baseline as Case Study~1. Because the metering gate and the descent evaluation are identical, the Baseline column of \Cref{tab:case2_champions} reproduces that of \Cref{tab:champions} exactly, and the two tables are directly comparable. The optimized-CDA saving falls from 17--23\% in the free descent to 6.5--15.5\% under the floors, and the optimized-DDA saving from 24--37\% to 20--34\%. The delayed deceleration keeps most of its advantage because it stays clean and fast to the same capture point, so its plan already sits higher over JAAJJ, whereas part of the CDA gain in Case Study~1 came from moving the capture point closer in, which the floors forbid.

\begin{figure*}[!t]
\centering
\begin{subfigure}{0.495\textwidth}
\includegraphics[width=\textwidth]{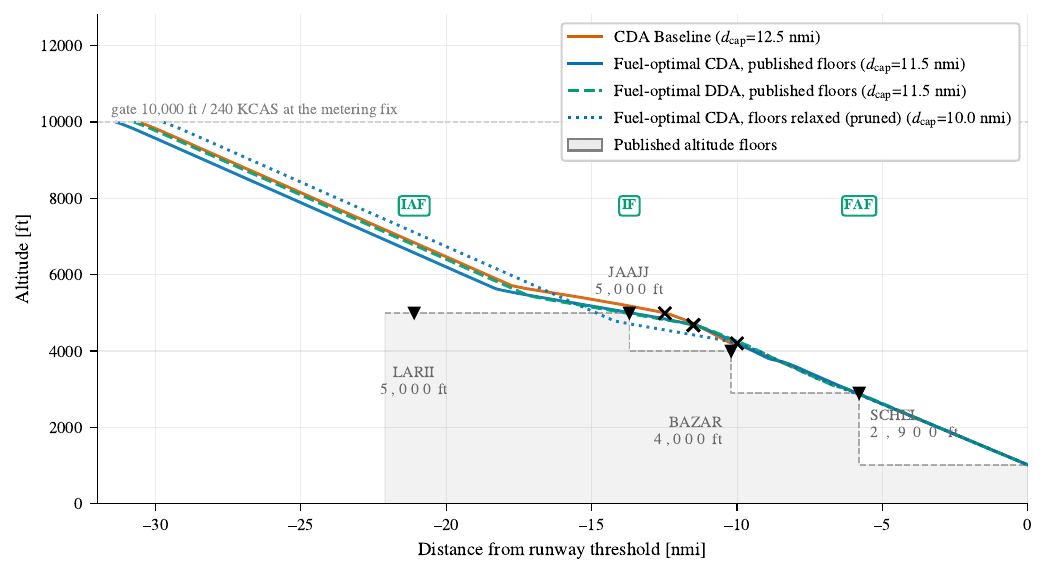}
\caption{B737-800, altitude}
\end{subfigure}\hfill
\begin{subfigure}{0.495\textwidth}
\includegraphics[width=\textwidth]{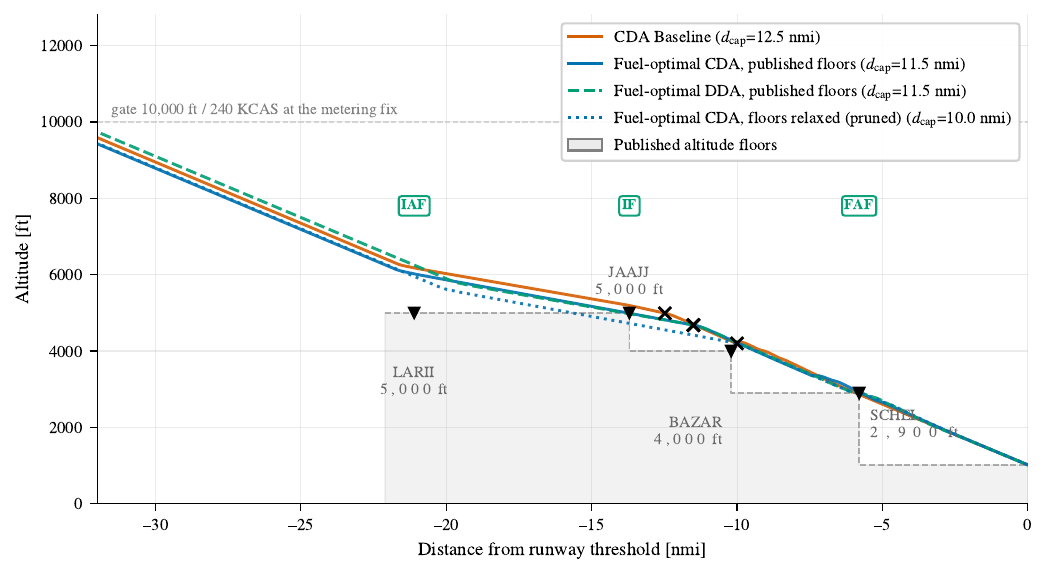}
\caption{A340-300, altitude}
\end{subfigure}\\[2pt]
\begin{subfigure}{0.495\textwidth}
\includegraphics[width=\textwidth]{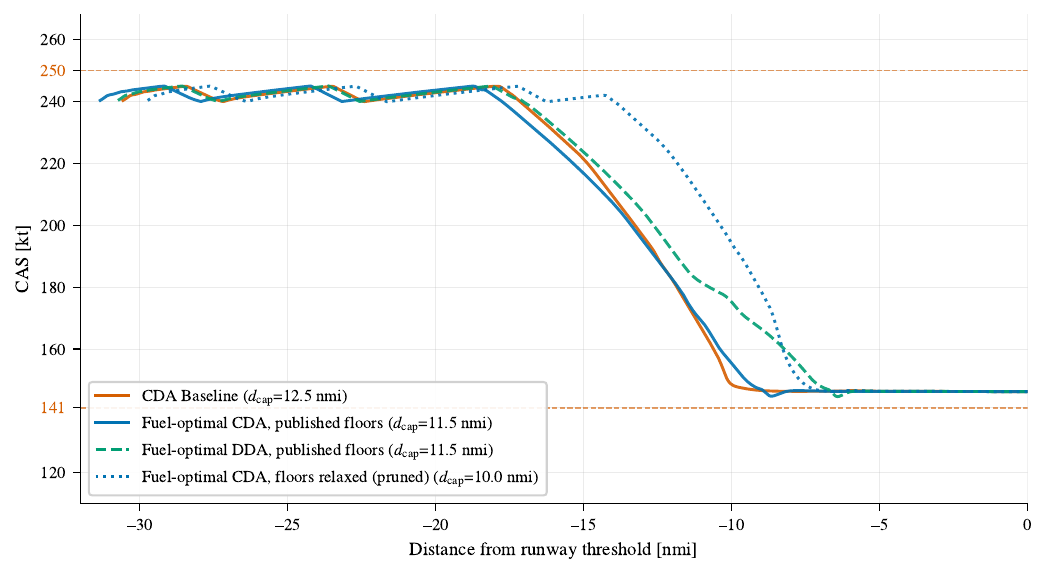}
\caption{B737-800, CAS}
\end{subfigure}\hfill
\begin{subfigure}{0.495\textwidth}
\includegraphics[width=\textwidth]{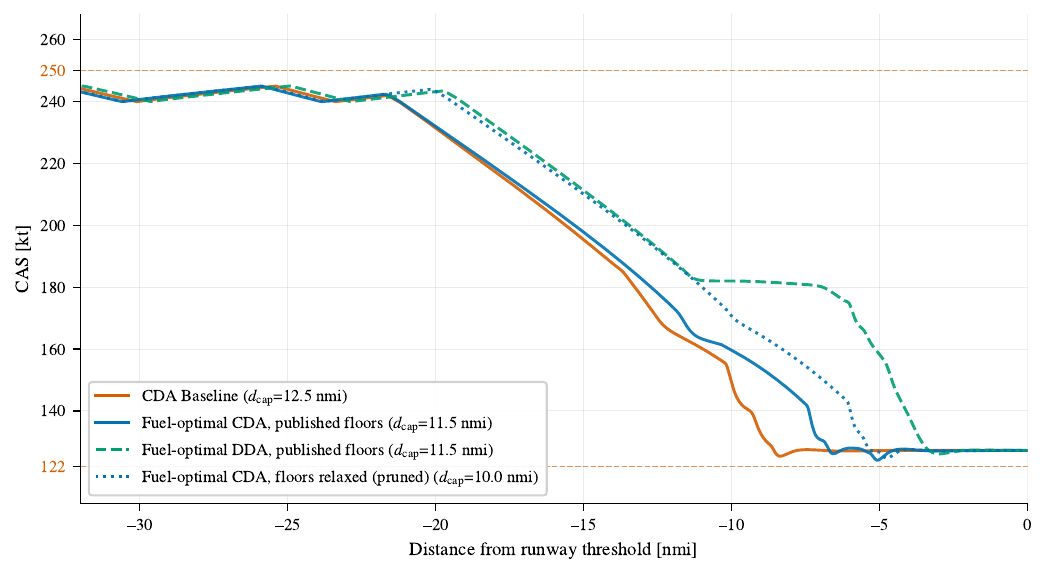}
\caption{A340-300, CAS}
\end{subfigure}
\caption{Simulated trajectories on the published final approach course at the zero-wind node for a Large type (B737-800, left) and a Heavy type (A340-300, right), showing the CDA Baseline, the fuel-optimal floor-feasible CDA, the fuel-optimal floor-feasible DDA, and the fuel-optimal CDA when the floors are relaxed. The shaded staircase is the charted altitude corridor of \Cref{tab:case2_floors}, triangles mark the constrained fixes, and crosses mark the glideslope captures. The floors-relaxed optimum captures at 10.0 nmi and crosses JAAJJ below the 5{,}000 ft floor, which is why the screen removes it. The red CAS ticks mark the 250 kt terminal limit and $V_{\mathrm{ref}}$.}
\label{fig:case2_profiles}
\end{figure*}

\begin{table}[t]
\centering
\caption{Fuel-optimal stabilized designs at the zero-wind node (fuel from the 10{,}000 ft metering gate to the threshold), with the charted floors enforced and with them relaxed. Savings are relative to the CDA Baseline ($V_{\min}{+}10$ triggers at the base capture distance), which is the same Baseline as \Cref{tab:champions}.}
\label{tab:case2_champions}
\begin{tabular}{lccccccccc}
\toprule
& \multicolumn{1}{c}{Baseline}
& \multicolumn{4}{c}{Charted floors enforced}
& \multicolumn{4}{c}{Floors relaxed}\\
\cmidrule(lr){3-6}\cmidrule(lr){7-10}
& & \multicolumn{2}{c}{CDA} & \multicolumn{2}{c}{DDA}
  & \multicolumn{2}{c}{CDA} & \multicolumn{2}{c}{DDA}\\
\cmidrule(lr){3-4}\cmidrule(lr){5-6}\cmidrule(lr){7-8}\cmidrule(lr){9-10}
Type & $F$ (kg) & $F$ (kg) & saving & $F$ (kg) & saving
     & $F$ (kg) & saving & $F$ (kg) & saving\\
\midrule
A319       & 135.1 & 114.2 & 15.5\% & 90.9  & 32.7\%
                   & 103.6 & 23.3\% & 84.8  & 37.2\%\\
B737-800   & 208.5 & 195.0 &  6.5\% & 166.8 & 20.0\%
                   & 173.3 & 16.9\% & 159.4 & 23.5\%\\
A340-300   & 579.5 & 501.2 & 13.5\% & 385.2 & 33.5\%
                   & 455.1 & 21.5\% & 371.3 & 35.9\%\\
B767-400ER & 441.8 & 390.9 & 11.5\% & 336.3 & 23.9\%
                   & 362.7 & 17.9\% & 332.7 & 24.7\%\\
\bottomrule
\end{tabular}
\end{table}

\begin{figure*}[!t]
\centering
\includegraphics[width=\textwidth]{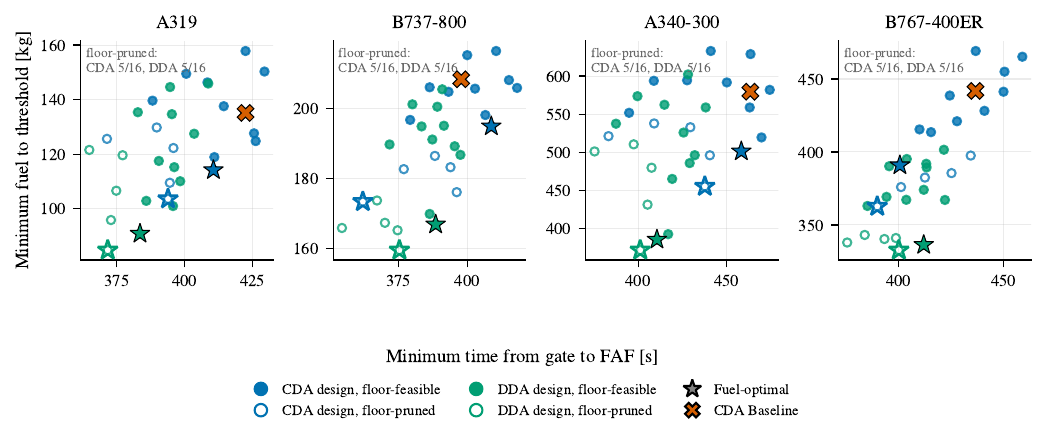}
\caption{Design menus at the zero-wind node, showing gate-to-FAF time versus gate-to-threshold fuel for every lattice design and airframe. Filled markers are floor-feasible designs and open markers are the designs the charted floors remove, which are exactly the 10.0 nmi capture row. Stars mark the fuel optima of each architecture and the cross is the CDA Baseline. The pruned designs occupy the fast and cheap corner of both clouds, so the floors remove the part of the menu that the unconstrained problem would select.}
\label{fig:case2_menu}
\end{figure*}

\subsection{Delay Authority Across the Arrival Flows}\label{subsec:case2_flows}

\Cref{fig:case2_flows} separates the fleet result by flow. Three patterns follow from the geometry of \Cref{tab:case2_flows}. The scheduler takes delay from the near side. Mean committed extensions are 7.6--10.9 nmi on MRCHH, STHRN, and JNGLE, compared with 1.5--2.6 nmi on SMELY, HAARY, and TIZZY. A nautical mile of far-side extension costs 1.9 track nmi and therefore roughly twice the fuel. Far-side flows still burn more fuel in absolute terms, 580--674 kg per aircraft compared with 325--515 kg at the Baseline. Surplus track miles account for 32--43\% of that fuel on the far side and 14--26\% on the near side. The co-optimized DDA saves a similar share of descent fuel on every flow, 22--29\%, but total-fuel savings range from 13.9\% on the far side to 24.0\% on STHRN, the flow with the least surplus track distance. Descent design cannot remove surplus track miles. Lowering far-side fuel would require a shorter route or a cheaper absorption altitude. A symmetric four-corner model averages out this behavior.

\begin{figure*}[!t]
\centering
\includegraphics[width=\textwidth]{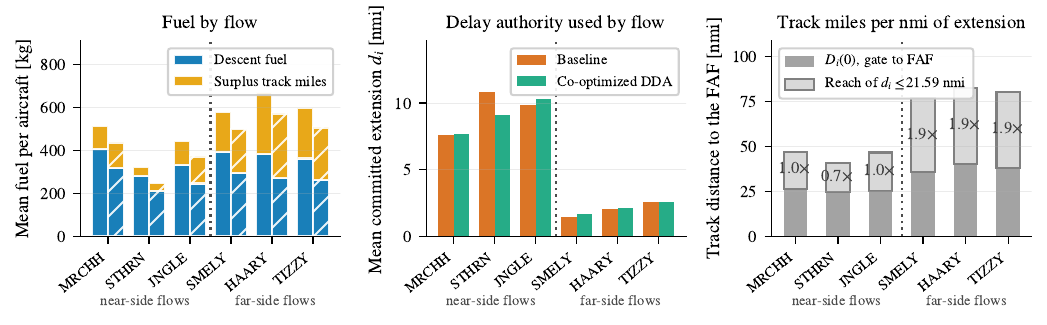}
\caption{Per-flow behavior on the six published arrival flows, grouped into near-side and far-side flows. The left panel splits the mean fuel per aircraft into descent fuel and the fuel spent on surplus track miles for the Baseline (left bar) and the co-optimized DDA (hatched right bar). The center panel shows the mean committed extension $d_i$ by flow. The right panel shows the track distance from the metering fix to the FAF at zero extension, with the additional reach of the extension authority $d_i \leq 21.59$ nmi stacked on top and annotated with the track miles spent per nautical mile of extension.}
\label{fig:case2_flows}
\end{figure*}

\subsection{Fleet Results Across Demand}\label{subsec:case2_fleet}

\Cref{fig:case2_timespace} shows the FAF landing timeline of one saturated 40 aircraft scenario. Under the Baseline the delay authority runs out late in the arrival push and four aircraft land with separation shortfalls totaling 267~s. The co-optimized DDA under CPS$_2$ lands the same traffic with a single 14 s shortfall and 8.3\% less fleet fuel, using the faster transits of the clean descent together with the re-sequenced slot order.

\begin{figure*}[!t]
\centering
\includegraphics[width=\textwidth]{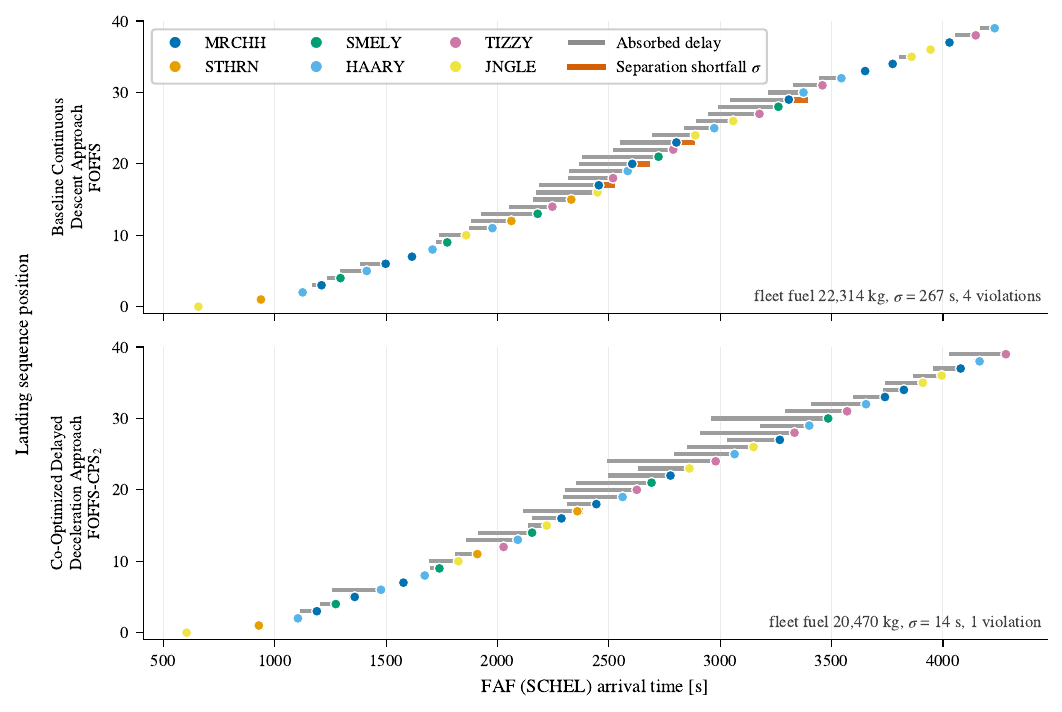}
\caption{FAF landing timeline for a saturated scenario (40 aircraft), comparing current practice (Baseline continuous descent under FOFFS, top) with the co-optimized delayed deceleration under FOFFS-CPS$_2$ (bottom). Each dot is an aircraft at its committed FAF time, colored by its arrival flow and stacked by landing sequence position. The gray bar behind each dot is the delay the aircraft absorbs relative to its unconstrained earliest arrival, and a red-orange bar marks any residual separation shortfall $\sigma$. The Baseline lands four aircraft with shortfalls totaling 267 s, whereas the co-optimized schedule absorbs the same demand with a single 14 s shortfall while burning 8.3\% less fleet fuel (20{,}470 kg against 22{,}314 kg).}
\label{fig:case2_timespace}
\end{figure*}

\Cref{fig:case2_sweep_metrics,fig:case2_sweep_savings} report the demand sweep and \Cref{tab:case2_sweep} tabulates selected levels. Three regimes appear, and they reproduce the Case-1 pattern on the published geometry.

Below the saturation onset the savings are flat and entirely design-driven. The co-optimized CDA saves 8.8--10.0\% and the co-optimized DDA 17.7--21.2\% against the Baseline, and CPS coincides with FOFFS because no slack has appeared yet. These fleet numbers sit below the design-level numbers of \Cref{tab:case2_champions} because every arm also pays for the surplus track miles that the published vectoring geometry imposes, and no descent design can remove them. Level flight at the metering altitude already accounts for 21.8\% of Baseline fleet fuel at the lowest demand level, and that share rises to 29.4\% at 36 AC/hr and 42.5\% at the top level, which sets the ceiling on what any descent-side intervention can achieve.

Saturation sets in between 36 and 42 AC/hr nominal, where the Baseline slack rises from 38.5 to 119~s per scenario. Above the onset the plain FOFFS saving erodes to 4.9\% for the CDA family and 10.2\% for the DDA family at the top level, because congestion forces level extension on every arm and the shared cost dilutes the descent-side advantage. This is the same mechanism as in Case Study~1. CPS re-sequencing recovers part of the saving and most of the safety margin. At the top level CPS$_2$ raises the paired saving to 6.3\% (CDA) and 12.3\% (DDA) and cuts the mean slack from 505 to 241~s and from 483 to 230~s, with violating aircraft falling from about 6.5 to 4.2--4.3 per scenario. CPS$_3$ trims the slack further, to 195 and 186~s, and gives back fuel, reflecting the lexicographic priority order. With safety ahead of fuel, the scheduler spends fuel to protect the schedule. The recovery is smaller than the free-descent case study shows, because the containment cap of \eqref{eq:containment} leaves less extension authority for CPS to exploit.

FEFS is the weak ordering, as in Case Study~1. Entry-order sequencing gives up 1.9 percentage points of saving at the lowest demand and up to 6.5 points near 24 AC/hr, and it carries 1.6 to 20 times the separation slack of FOFFS at every congested level. At the top two levels its paired fuel saving crosses above plain FOFFS, at 9.6\% (CDA) and 14.8\% (DDA), while its slack stays 1.6 times higher and its violation count reaches 9.7 aircraft per scenario. The two orderings are comparable only through the slack panel. FOFFS spends fuel to protect the schedule where FEFS violates it.

\begin{figure*}[!tp]
\centering
\includegraphics[width=0.95\textwidth]{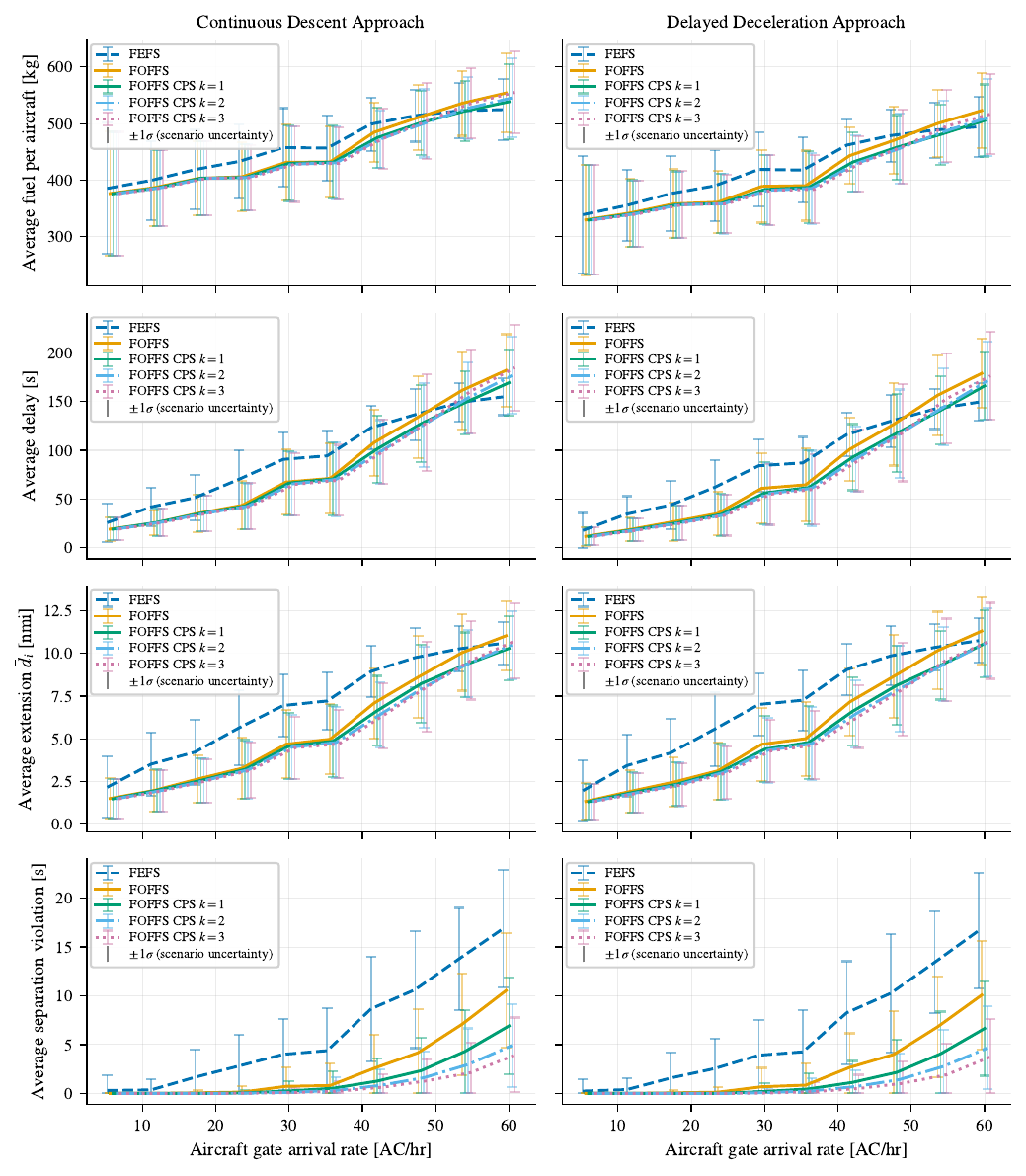}
\caption{Average per-aircraft fuel burn, average delay, average centerline extension $\bar d_i$, and average separation violation (rows) versus aircraft gate arrival rate, aggregated over the six published flows, for the CDA family (left column) and the DDA family (right column) under FEFS, FOFFS, and FOFFS-CPS$_{1,2,3}$ (50 seeds per demand level). Within each row the two columns share the vertical axis, so a horizontal read gives the difference between the two architectures directly. Error bars show the $\pm 1\sigma$ scenario (demand and wind) uncertainty, and lower bars are clipped at the physical zero.}
\label{fig:case2_sweep_metrics}
\end{figure*}

\begin{figure*}[!t]
\centering
\includegraphics[width=\textwidth]{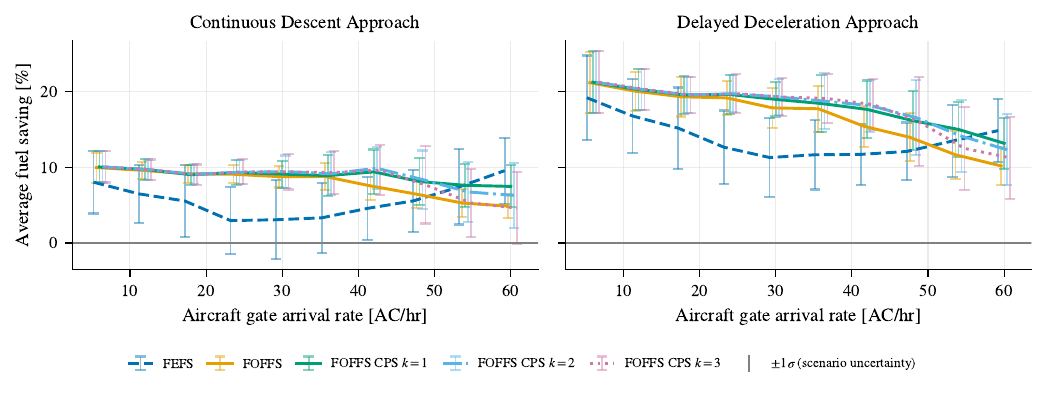}
\caption{Average fleet-fuel saving relative to the CDA Baseline, paired per scenario, versus aircraft gate arrival rate on the six published flows, for the CDA family (left) and the DDA family (right) under FEFS, FOFFS, and FOFFS-CPS$_{1,2,3}$ on a shared vertical scale. Error bars show the $\pm 1\sigma$ scenario uncertainty.}
\label{fig:case2_sweep_savings}
\end{figure*}

\begin{table*}[t]
\centering
\caption{Case Study 2 demand-sweep summary (50 seeds per level, common random numbers). The Baseline is current practice (CDA with $V_{\min}{+}10$ triggers at the base capture distance) and savings are paired per scenario against it. The floors-relaxed arms repeat the FOFFS commitment with the charted floors removed. $\sigma$ is the mean total separation slack per scenario and $\bar n_{\mathrm{viol}}$ the mean number of violating aircraft. Column headings give the nominal aggregate gate arrival rate in AC/hr.}
\label{tab:case2_sweep}
\begin{tabular}{clcccccccc}
\toprule
& & \multicolumn{4}{c}{Fuel saving over Baseline (\%)}
& \multicolumn{2}{c}{$\sigma$ (s)}
& \multicolumn{2}{c}{$\bar n_{\mathrm{viol}}$}\\
\cmidrule(lr){3-6}\cmidrule(lr){7-8}\cmidrule(lr){9-10}
& Policy & 12 & 36 & 48 & 60 & 48 & 60 & 48 & 60\\
\midrule
\multicolumn{2}{l}{Baseline (current practice)}
                & --- & --- & --- & --- & 194.6 & 524.6 & 3.44 & 6.84\\
\midrule
\multirow{6}{*}{CDA}
& FEFS          &  6.5 &  3.3 &  5.6 &  9.6 & 429.3 & 794.6 & 6.54 & 9.70\\
& FOFFS         &  9.6 &  8.8 &  6.4 &  4.9 & 176.3 & 504.7 & 3.18 & 6.48\\
& FOFFS-CPS$_1$ &  9.7 &  8.9 &  8.1 &  7.5 & 100.5 & 336.3 & 2.26 & 5.14\\
& FOFFS-CPS$_2$ &  9.7 &  9.1 &  8.4 &  6.3 &  69.9 & 240.9 & 1.70 & 4.20\\
& FOFFS-CPS$_3$ &  9.7 &  9.3 &  7.7 &  4.6 &  55.1 & 194.7 & 1.60 & 4.00\\
& FOFFS, floors relaxed
                & 15.0 & 13.8 & 10.2 &  7.5 & 163.2 & 487.5 & 2.98 & 6.12\\
\midrule
\multirow{6}{*}{DDA}
& FEFS          & 16.8 & 11.7 & 12.1 & 14.8 & 415.2 & 785.4 & 6.08 & 9.78\\
& FOFFS         & 20.1 & 17.7 & 14.0 & 10.2 & 169.0 & 482.9 & 3.02 & 6.60\\
& FOFFS-CPS$_1$ & 20.3 & 18.4 & 16.2 & 13.2 &  93.4 & 324.3 & 2.08 & 5.04\\
& FOFFS-CPS$_2$ & 20.3 & 18.7 & 16.6 & 12.3 &  62.0 & 229.6 & 1.36 & 4.34\\
& FOFFS-CPS$_3$ & 20.3 & 19.1 & 16.1 & 11.2 &  44.9 & 186.4 & 1.26 & 3.92\\
& FOFFS, floors relaxed
                & 23.0 & 20.4 & 16.0 & 11.6 & 159.7 & 470.6 & 2.82 & 6.48\\
\bottomrule
\end{tabular}
\end{table*}

\subsection{The Cost of the Published Crossing Restrictions}\label{subsec:case2_cost}

The floors-relaxed arms price the published structure on identical traffic, winds, and geometry. \Cref{fig:case2_floors} reports the two saving curves and their difference. Below saturation the charted floors cost 6.7\% of fleet fuel for the CDA family and 3.6\% for the DDA family. The penalty falls to 2.8\% and 1.6\% at the top demand level, because sequencing rather than descent design then dominates the fuel, and the level extension that congestion forces on every arm is unaffected by the floors. The DDA pays less than the CDA throughout, for the reason identified in \Cref{subsec:case2_designspace}. Its clean profile already clears JAAJJ, so it loses less when the low capture row is removed.

Placing the two case studies side by side isolates what the published structure costs at the two levels of the problem. At the design level the charted floors remove 31\% of the lattice, pin the glideslope capture to the charted platform, and reduce the optimized-CDA saving from 17--23\% to 6.5--15.5\%. At the fleet level the co-optimized savings below saturation fall from about 15.5\% (CDA) and 23\% (DDA) in the free-descent environment to about 9.6\% and 20\% on the published flows. The floors-relaxed arms attribute that gap. With the floors removed, published-flow savings recover to 15.0\% and 23.0\%, matching the free-descent values. The fleet-level gap between the two case studies comes from the altitude restrictions, not the vectoring geometry. The geometry instead changes how the benefit distributes across the flows, as \Cref{subsec:case2_flows} shows, and how much fuel each flow burns in absolute terms.

The comparison also shows which decision matters at each demand level. Below the wake capacity that decision is the descent architecture, and the delayed deceleration is worth roughly twice the optimized continuous descent on the same final approach angle. At capacity the sequencing decision and the placement of the extension matter more, because together they keep the delivered demand inside the delay authority the airspace actually has. These conclusions rest on the six-flow abstraction of the vector segment and on the uniform per-flow demand of the sweep, and we return to both in \Cref{sec:discussion}.

\begin{figure*}[!t]
\centering
\includegraphics[width=\textwidth]{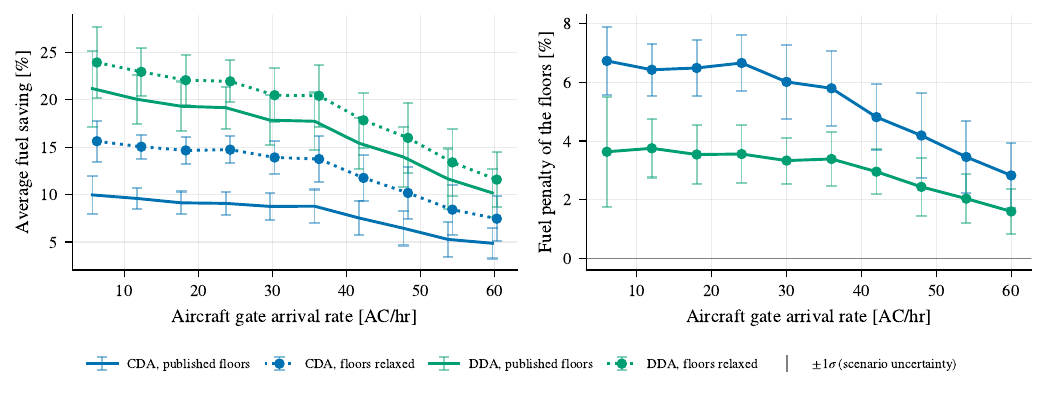}
\caption{Cost of the charted crossing restrictions. The left panel shows the average paired fleet-fuel saving against the Baseline for the CDA and DDA families under FOFFS, with the charted floors enforced (solid) and relaxed (dotted). The right panel shows the difference between the two curves, which is the fuel penalty of the published floors. Error bars show the $\pm 1\sigma$ scenario uncertainty.}
\label{fig:case2_floors}
\end{figure*}

\section{Discussion}\label{sec:discussion}

The individual findings of the two case studies combine into one pattern: the value of each decision in the co-optimization depends on which constraint binds. Below saturation the binding constraint is the descent design itself, and most of the benefit comes from the vertical decision set. At saturation the runway binds, and the sequencing decision and the placement of the extension carry the remaining margin. In the published-flow environment the charted structure binds before either, removing the part of the design space the optimizer would use. Wind, finally, binds the individual aircraft but not the fleet, because the wind-aware plan absorbs it before it reaches the schedule. The subsections take these results in turn, relate them to prior studies, and state what each implies for deployment.

\subsection{The Dominant Decision at Each Demand Level}\label{subsec:disc_leverage}

Read together, the two case studies show that the dominant decision changes with the demand level. Below the wake capacity, the descent design carries almost the entire benefit. The saving is flat across demand levels, matches the design-level menus of \Cref{tab:champions,tab:case2_champions}, and owes nothing to sequencing, because no separation slack has appeared. At and above saturation, the descent design still helps, but the level extension that congestion forces on every arm dilutes it. In the published-flow case study, level flight at the metering altitude accounts for 22\% of fleet fuel at the lowest demand level and 43\% at the highest (\Cref{subsec:case2_fleet}), and no descent design can recover fuel burned on that segment. The remaining margin lies in the sequencing decision and in where the extension is flown. Bounded position shifting with $k = 2$ recovers 1.4 percentage points of fuel saving for the CDA family and 2.1 for the DDA family at the top demand level, and it more than halves the separation slack.

The split has a practical reading. A fuel-efficient arrival procedure delivers its full design-level benefit only while the runway is not the binding constraint. Once it is, a scheduler that can re-pair aircraft within two positions has more leverage than further refinement of the descent profile. The two interventions complement each other, since CPS preserves the descent-side saving above saturation instead of replacing it.

This reading is consistent with the earlier coupled studies and extends a single-aircraft result to the fleet. \citet{robinson2010benefits} found the continuous descent benefit in dense traffic limited by the spacing buffer the procedure requires, and the erosion of the plain-FOFFS saving above saturation in both case studies is the same mechanism measured against a common baseline. \citet{nikoleris2016comparison} showed for individual aircraft that absorbing delay by slowing the descent is cheaper than flying extra level distance. The joint commitment applies that ordering fleet-wide under separation constraints, and the 58--65\% of aircraft that keep their unconstrained fuel-optimal design in Case Study~1 measure how often the separation constraints leave the cheap option available.

\subsection{Implications of the Matched-Angle Comparison}\label{subsec:disc_matched}

Fixing a common $3.0^\circ$ final changes how the DDA literature reads. On the matched final, DDA saves 8\% relative to an optimized CDA for the B737-800 and B767-400ER and 18\% for the A319 and A340-300, and the advantage tracks clean-to-approach speed range rather than weight class. The noise result is consistent with this reading. With the glide path fixed, the two co-optimized architectures produce 55 dBA footprints within a few percent of each other, and the unsigned difference map reflects different committed extensions rather than a systematic altitude difference. The 4--8 dB reductions under the flight track reported by \citet{thomas2021modeling} were measured on procedures that combined delayed deceleration with a steeper final as one design change. The matched-angle result here indicates that the quieter descent comes substantially from the steeper final rather than from delayed deceleration alone, while the fuel advantage survives the separation. Because descent angle changes fuel by itself \cite{turgut2019effects}, the two design changes should be reported separately. The companion study reports them separately at matched final angles and finds that the glideslope angle carries the larger expected-fuel saving for a single aircraft \cite{pang2026optimal}.

The direction of the noise result also matters. Both fuel-optimal families replace the 5{,}000 ft platform with a continuous idle path that captures closer in, so each descent runs a few hundred feet lower over the 10--15 nmi band, and the exposed area above 55 dBA grows by 3--11\% relative to current practice at 40 AC/hr. Fuel optimization on a fixed glide path therefore carries a small ground-noise penalty. The penalty is larger for the CDA family than for the DDA family, so delayed deceleration attains its fuel saving close to the Baseline footprint on this metric. Deployment decisions should weigh fuel against noise rather than assume that the two improve together.

\subsection{Wind Uncertainty at the Aircraft and Fleet Levels}\label{subsec:disc_wind}

Wind is the largest single influence on an individual descent and one of the smallest influences on the fleet metrics. The per-aircraft sensitivity confirms what \citet{park2016vertical} found for a single descent, that altitude-dependent winds are a first-order input to the vertical profile. The fleet-level cancellation is the observation this study adds. Each aircraft re-anchors its descent plan to the wind it observes at the metering gate (\Cref{subsec:case1_wind}), so the residual wind effect is a per-aircraft perturbation that averages out across the 5 to 46 aircraft of a scenario. A demand realization, by contrast, shifts the whole scenario through stream bunching and the weight-class mix. The error bars in \Cref{fig:sweep_metrics,fig:case2_sweep_metrics} therefore mainly measure demand uncertainty, and wind must be handled inside the descent plan rather than at the fleet level.

The same structure explains why the stabilized-approach chance constraint never became active. With the wind observed and the plan re-anchored, stabilization held as a hard constraint at every quadrature node in both case studies, and only the wide-body DDA lost designs at the strongest tailwind nodes. A static procedure design faces a harder problem: a single wind-blind design committed against the climatology can guarantee stabilization only in probability, and the published flows produce infeasible tail cells for such a design. The value of a per-aircraft design menu would grow further if wind observation were delayed.

\subsection{Implications for Automation}\label{subsec:disc_automation}

The framework moves the expensive computation offline. The design lattice is evaluated once per aircraft type, architecture, design, and wind node (\Cref{subsec:properties}), and the fleet layer reads the cache. Per-aircraft commitment reduces to enumeration over at most 16 designs, with a bisection for the smallest feasible extension, and each aircraft is touched once. Case Study~2 reuses the Case-1 cache without new simulation because the two studies share the metering-gate condition. An operational implementation would follow the same split: the lattice would be certified per fleet type in advance, and a lookup-driven scheduler would run as aircraft cross the metering ring.

The single-pass commitment also produces a usable output. For each aircraft, the scheduler emits a glideslope-capture distance, flap trigger speeds, and an extension distance. These quantities map onto instructions a controller can issue and a flight management system can accept, which is why the framework schedules simulated trajectories rather than travel-time envelopes.

The difference from a target landing time matters most when no crew reconciles the ground plan with the aircraft. A target time leaves the aircraft to decide how to meet it, and a vector rests on the assumption that the aircraft can comply; pilots check both today. In an autonomous arrival, the check must be part of the commitment. The framework therefore verifies every candidate against the stabilized-approach criteria before selection, and the committed object names the configuration changes as well as the times and speeds. The same property makes the commitment auditable, because every selected trajectory exists in the cache with its wind node and stabilization verdict. Taken with \cite{pang2026trajectory}, the framework covers the ground assignment and airborne execution layers in one place.

\subsection{Limitations and Future Studies}\label{subsec:disc_limitations}

The limitations of this study also mark out the extensions we consider most useful, so we state each limitation together with the study it motivates.

First, the vector segment downstream of the metering fix is a tangent leg, an RF arc, and a centerline extension, in parity with the vector-to-final model of the lateral study. It routes the far-side flows on a direct diagonal to the intercept instead of along the charted downwind, which is 20--25 nmi longer, and it omits the flow-specific step-downs upstream of the intercept because it replaces the charted transition there. An explicit model would raise far-side fuel, reduce far-side delay authority, and prune more of the lattice than we report. Repeating Case Study 2 with the charted downwind and the step-down restrictions in place would price the extension on the path the aircraft actually flies and complete the cost of the published structure.

Second, per-flow demand is uniform across the six flows, whereas real A80 demand is heavier from the north and the west. Because the track-mile price of extension differs by a factor of two between the near-side and the far-side flows, an asymmetric pattern would change how much extension the scheduler can find, in a direction that depends on which flows carry the load. Replaying recorded per-flow demand, with recorded winds in place of the climatology quadrature, would test whether the near-side flows retain enough delay authority when they also carry most of the traffic.

Third, the wind-observable setting studied here is the favorable limit, since each aircraft commits against the wind it observes at the gate. A static design problem, in which a single wind-blind procedure is published against the climatology, would turn the stabilized-approach requirement into a chance constraint with real probabilistic content, and \Cref{subsec:disc_wind} indicates that the published flows would produce infeasible tail cells for such a design. Comparing the per-aircraft design menu against a single robust published design would quantify the value of observing the wind before commitment.

Fourth, the architecture is fixed per experiment, so the fleet is homogeneous in CDA or DDA, and same-flow non-overtaking precedence is not enforced, in parity with the lateral study. A mixed fleet is the case an operator would face first, because delayed-deceleration capability will not arrive at every aircraft at once. The speed differential between a clean arrival at 240 kt and a configured arrival at 170 kt on the same final approach course raises a spacing problem that the homogeneous-fleet formulation does not address.

Fifth, the noise model interpolates noise-power-distance curves from the EUROCONTROL database at a fixed engine state, substitutes curves for airframes without a native entry, and omits the airframe-noise increment of late flap deployment, so the reported footprints isolate the geometric effect of altitude and path and are not certification-grade contours. A configuration-aware noise model would let the acoustic effect of the flap schedule enter the fuel-noise comparison of \Cref{subsec:disc_matched}.

Lastly, all delay is absorbed inside the terminal area. There is no holding variable, any residual shortfall is recorded as separation slack rather than pushed upstream, and in practice the center descent controller meters the inbound flow before it reaches the terminal boundary, so coupling the terminal co-optimization with that upstream metering is a natural extension. A runtime characterization in the manner of the lateral study would complete the operational picture by separating the offline cost of certifying the lattice per fleet type from the online per-aircraft enumeration, since only the latter has to keep pace with the arrival stream.

\section{Conclusion}\label{sec:conclusion}

This paper formulates terminal arrival management as a four-dimensional simulation-in-the-loop co-optimization problem where each aircraft's vertical descent design and lateral path extension are committed together, extending a previous work of lateral scheduler that optimized only controller-issued actions \cite{pang2026trajectory} and a companion study that optimized and verified the single-aircraft descent designs in physics-based simulation \cite{pang2026optimal}. The vertical decision is a finite lattice over glideslope-capture distance and flap-deployment trigger speeds, and each lattice point is evaluated offline by a wind-aware backward plan and a six-degree-of-freedom forward simulation that returns descent time, fuel burn, minimum descent track length, and stabilized-approach feasibility. Minimum descent track length couples the vertical design to the lateral path, so the scheduler absorbs required delay inside the descent instead of on level track miles. A rolling-horizon forward pass commits each aircraft once and resolves safety, fuel, and delay in strict priority order under first-on-final-first-serve sequencing with bounded position shifting and pair-specific wake separation, and online scheduling reduces to table lookups against the simulation cache.

Two Atlanta Large TRACON case studies quantify the benefit of the proposed approach. In a free-descent environment, the co-optimized continuous descent saves about 15\% of fleet fuel below saturation and delayed deceleration saves about 23\%, and bounded position shifting with two positions of freedom holds these savings nearly flat through saturation while cutting mean separation shortfall by roughly a factor of four. On the six published Runway 8L arrival flows, charted final-approach floors remove 31\% of the design lattice, reduce the continuous-descent saving to 9--10\%, and leave the delayed-deceleration saving at 20--21\%; relaxing the floors prices them at 1.6--6.7\% of fleet fuel depending on architecture and demand. 

On a matched $3.0^\circ$ final, the delayed-deceleration advantage over an optimized continuous descent is 8--18\% and tracks clean-to-approach speed range, and the two architectures produce ground-noise footprints within a few percent of each other, which indicates that quieter descents reported in the delayed-deceleration literature come substantially from the steeper final flown with that procedure. Wind changes individual-descent fuel by 34--81\% across the climatology but explains at most 4\% of fleet-metric variance, and delay authority concentrates in the near-side flows.

For each aircraft, the framework commits a glideslope-capture distance, a flap trigger schedule, and an extension distance, each verified in simulation against stabilized-approach criteria at the observed wind. In operations, a flight management system can execute this object directly, and ground automation can check it without a crew in the loop.

\section*{Acknowledgment}
This work was supported by the National Aeronautics and Space Administration (NASA) University Leadership Initiative (ULI) program under project ``Autonomous Aerial Cargo Operations at Scale,'' via grant No.~80NSSC21M071 to the University of Texas at Austin. Any opinions, findings, conclusions, or recommendations expressed in this material are those of the authors and do not necessarily reflect the views of the project sponsor. The authors also thank Dr. Jiacheng Xie of the Georgia Institute of Technology, an instrument-rated pilot, and Jim Allerdice, the chief designer of the Area Navigation (RNAV) infrastructure at A80 TRACON, for helpful discussions regarding charted procedures at KATL. The authors also thank EUROCONTROL for providing the ANP data.

\appendix

\section{Monotonicity Lemmas}\label{app:lemmas}

This appendix states the two monotonicity results used in \Cref{subsec:geometry,subsec:properties}. \Cref{lem:monotone_D} concerns the lateral map \eqref{eq:Dtot} alone, and \Cref{lem:monotone_tf} transfers its monotonicity to the FAF arrival time and the total fuel through the linking equations \eqref{eq:time_link} and \eqref{eq:fuel_link}.

\begin{lemma}[Monotone track distance]\label{lem:monotone_D}
For every entry point and turn direction in this study, $D_i$ is continuously differentiable and strictly increasing on $[0, d_{\max}]$. On a uniform grid of $2001$ points its derivative satisfies $0.28 \le D_i'(d_i) \le 1.98$ for the four corner geometries of Case Study 1 and the six published flows of Case Study 2.
\end{lemma}

\begin{lemma}[Monotone time and fuel]\label{lem:monotone_tf}
For a fixed design $\chi$, both the FAF arrival time \eqref{eq:time_link} and the total fuel \eqref{eq:fuel_link} are continuous and strictly increasing in the extension $d$ on the feasible interval $[d_{\mathrm{floor}}(\chi), d_{\max}]$, and their common derivative structure is
\begin{equation}
\frac{\partial t_i}{\partial d}
  = \frac{3600}{V^{\mathrm{gs}}_i}\, D_i'(d),
\qquad
\frac{\partial F_i}{\partial d}
  = K_\gamma \frac{V^{\mathrm{TAS}}_\gamma}{V^{\mathrm{gs}}_i}\, D_i'(d),
\label{eq:mono_derivs}
\end{equation}
so that $\partial F_i/\partial d = (K_\gamma V^{\mathrm{TAS}}_\gamma/3600)\,\partial t_i/\partial d$.
\end{lemma}

\begin{proof}
The surplus distance is $\Delta_i(d) = D_i(d) - \ul D$, which is strictly increasing by \Cref{lem:monotone_D} and nonnegative on $[d_{\mathrm{floor}}(\chi), d_{\max}]$ by the definition of $d_{\mathrm{floor}}$. Both $t^{\mathrm{des}}$ and $F^{\mathrm{des}}$ are constants of the design at the observed node, so \eqref{eq:time_link} and \eqref{eq:fuel_link} are affine in $\Delta_i$ with positive slopes $3600/V^{\mathrm{gs}}_i$ and $K_\gamma V^{\mathrm{TAS}}_\gamma/V^{\mathrm{gs}}_i$. Differentiating gives \eqref{eq:mono_derivs}, and the ratio of the two slopes is independent of $d$.
\end{proof}

\section{Structural Propositions}\label{app:props}

This appendix states and proves the structural results of \Cref{subsec:properties}. Throughout, fix an aircraft $i$ of type $\gamma$, an architecture $a$, and an observed wind node $w_i$, with the cached quantities of \Cref{subsec:properties} and the stabilized menu $\mc S = \{\chi \in \mc X^{a}_{\gamma} : s(\chi, w_i) = 1\}$.

\begin{prop}[Exact per-aircraft commitment]\label{prop:exact}
Let $\mc S = \{\chi \in \mc X^{a}_\gamma : s(\chi, w_i) = 1\}$ be the stabilized menu, and for $\chi \in \mc S$ define
\begin{equation}
d_{\mathrm{floor}}(\chi) = \min\{d \in [0, d_{\max}] :
  D_i(d) \ge \ul D\}, \quad
d_\chi = \min\{d \ge d_{\mathrm{floor}}(\chi) :
  t_i(\chi, d) \ge t^{\mathrm{req}}_i\},
\label{eq:dchi}
\end{equation}
with $d_\chi = d_{\max}$ when the second set is empty. Then $(\chi, d_\chi)$ is the fuel-minimal feasible commitment for design $\chi$, both minima in \eqref{eq:dchi} are attained and unique, and each is computed by bisection on a strictly increasing continuous function. The optimum of \eqref{eq:joint_commit} is attained at $\arg\min_{\chi \in \mc S} F_i(\chi, d_\chi)$, so enumerating $|\mc S| \le |\mc X^{a}_\gamma|$ designs and bisecting twice per design solves the problem exactly.
\end{prop}

\begin{proof}
Both sets in \eqref{eq:dchi} are the superlevel sets of continuous strictly increasing functions of $d$ by \Cref{lem:monotone_D,lem:monotone_tf}, so each is a closed interval of the form $[d^\ast, d_{\max}]$ and its minimum is attained and unique. Bisection converges to $d^\ast$ because the defining inequality is monotone in $d$. Fix $\chi$ and let $d$ be any feasible extension for it, so $d \ge d_{\mathrm{floor}}(\chi)$ and either $t_i(\chi, d) \ge t^{\mathrm{req}}_i$ or $d = d_{\max}$. In both cases $d \ge d_\chi$, and $F_i(\chi, \cdot)$ is strictly increasing by \Cref{lem:monotone_tf}, hence $F_i(\chi, d) \ge F_i(\chi, d_\chi)$ with equality only at $d = d_\chi$. The feasible set of \eqref{eq:joint_commit} is therefore the finite set $\{(\chi, d_\chi) : \chi \in \mc S\}$ up to fuel-dominated points, and minimizing over it is a finite comparison.
\end{proof}

For \Cref{prop:lex}, write $\ol F$ and $\ul F$ for the largest and smallest total fuel over the stabilized menu at $d_{\max}$ and at the feasibility floor respectively, and $\ol t$, $\ul t$ for the corresponding arrival times.

\begin{prop}[Lexicographic equivalence]\label{prop:lex}
Suppose the weights satisfy
\begin{equation}
W_{\mathrm{safe}} > \frac{W_{\mathrm{fuel}}(\ol F - \ul F)
  + W_{\mathrm{delay}}(\ol t - \ul t)}{\varepsilon_\sigma},
\qquad
W_{\mathrm{fuel}} > \frac{W_{\mathrm{delay}}(\ol t - \ul t)}
  {\varepsilon_F},
\label{eq:weightcond}
\end{equation}
where $\varepsilon_\sigma$ and $\varepsilon_F$ are the smallest nonzero gaps between attainable slack values and between attainable fuel values over the menu. Then the minimizer of the weighted per-aircraft objective in \Cref{alg:commit} coincides with the lexicographic minimizer of $(\sigma_\chi,\, F_i(\chi, d_\chi),\, t_i(\chi, d_\chi))$.
\end{prop}

\begin{proof}
The candidate set is finite by \Cref{prop:exact}, so the attainable values of each term form a finite set and the gaps $\varepsilon_\sigma, \varepsilon_F > 0$ exist. If two candidates differ in slack, the difference is at least $\varepsilon_\sigma$, and the first inequality in \eqref{eq:weightcond} makes $W_{\mathrm{safe}}\varepsilon_\sigma$ exceed the largest possible opposing contribution of the fuel and delay terms, so the smaller-slack candidate wins. Among candidates with equal slack the same argument applies to the fuel term with the second inequality, and among candidates equal in slack and fuel the delay term decides.
\end{proof}

For \Cref{prop:ws}, let $\boldsymbol w = (w_1,\dots,w_N)$ be the vector of per-aircraft wind nodes, drawn independently from the discretized climatology, and let $J^\star(\boldsymbol w)$ be the optimal value of \eqref{eq:fleet_problem} at that realization.

\begin{prop}[Wait-and-see interpretation]\label{prop:ws}
Because every decision variable of \eqref{eq:fleet_problem} is chosen after the corresponding wind node is observed, the reported fleet metrics estimate $\mathbb E_{\boldsymbol w}[J^\star(\boldsymbol w)]$, which is the wait-and-see value of the stochastic program. For any wind-blind design rule $\pi$ that must commit $\chi_i$ before observing $w_i$,
\begin{equation}
\mathbb E_{\boldsymbol w}\bigl[J^\star(\boldsymbol w)\bigr]
  \;\le\; \mathbb E_{\boldsymbol w}\bigl[J^\pi(\boldsymbol w)\bigr].
\label{eq:wsbound}
\end{equation}
\end{prop}

\begin{proof}
Pointwise in $\boldsymbol w$, the wind-blind commitment of $\pi$ is feasible for \eqref{eq:fleet_problem} whenever it is stabilized, so $J^\star(\boldsymbol w) \le J^\pi(\boldsymbol w)$ at every realization at which $\pi$ is feasible, and $J^\pi = +\infty$ elsewhere. Taking expectations gives \eqref{eq:wsbound}.
\end{proof}

\bibliography{ref}

\end{document}